\documentclass[english,12pt]{article}

\usepackage{times}

\usepackage[utf8]{inputenc}
\usepackage{url}
\usepackage{color}
\usepackage{amsmath,amsthm,amssymb,amsfonts,latexsym,hyperref,multirow,verbatim}
\usepackage{lineno,setspace} 
\usepackage{graphicx}
\usepackage{floatrow}
\usepackage{subfigure}
\usepackage{authblk}
\usepackage{geometry}
\usepackage{lipsum}
\usepackage{natbib}
\usepackage{cleveref}
\usepackage{algorithm,algpseudocode}
\usepackage[normalem]{ulem}
\usepackage{pbox}
\usepackage{slashbox}

\usepackage{tikz}
\usepackage{verbatim}
\usetikzlibrary{arrows,shapes}
\usetikzlibrary{decorations.pathmorphing}
\usetikzlibrary{decorations.pathreplacing}
\usetikzlibrary{decorations.shapes}
\usetikzlibrary{decorations.text}
\usetikzlibrary{decorations.markings}
\usetikzlibrary{decorations.fractals}
\usetikzlibrary{decorations.footprints}
\usepackage{sectsty}
\usepackage{titlecaps}

\newcommand{\R}{\mathbb{R}} 
 
\newcommand{\pr}{\mathbb{P}} 
\newcommand{\prt}{\mathbb{P}_{\theta}}
\newcommand{\prtl}{\mathbb{P}_{\theta_l}}
\newcommand{\prphi}{\mathbb{P}_{\phi}}
\newcommand{\prpsi}{\mathbb{P}_{\psi}}
\newcommand{\prpsil}{\mathbb{P}_{\psi_l}}
\newcommand{\esp}{\mathbb{E}}

\newcommand{\argmin}{\mathop{\mathrm{Argmin}}}
\newcommand{\Xil}{X_i^{l}}
\newcommand{\xil}{x_i^{l}}
\newcommand{\Xjl}{X_j^{l}}
\newcommand{\xtil}{\tilde{x}_i^{l}}
\newcommand{\xtjl}{\tilde{x}_j^{l}}
\newcommand{\Yil}{Y_i^{l}}
\newcommand{\Yjl}{Y_j^{l}}
\newcommand{\A}{\textbf{A}}
\newcommand{\Al}{A^{l}}
\newcommand{\Aijl}{A_{i,j}^{l}}
\newcommand{\bX}{\mathbf{X}}

\newcommand{\btx}{\tilde{\mathbf{x}}}
\newcommand{\bY}{\mathbf{Y}}
\newcommand{\Xmil}{X_{-i}^{l}}
\newcommand{\Ni}{\mathcal{N}^\star_i}
\newcommand{\XNvl}{X_{\Ni}^{l}}

\newcommand{\xtNvl}{\tilde{x}_{\Ni}^{l}}

\newcommand{\psit}{\psi^{(t)}}
\newcommand{\psilt}{\psi_l^{(t)}}

\newcommand{\betap}{\beta_{l,co-pres}}
\newcommand{\betaa}{\beta_{l,co-abs}}
\newcommand{\betaplp}{\beta_{l',co-pres}}
\newcommand{\betaalp}{\beta_{l',co-abs}}

\newcommand{\degi}{\deg_{G^{\star}}(i)}
\newcommand{\degj}{\deg_{G^{\star}}(j)}
\newcommand{\1}{\boldsymbol{1}}

\newcommand{\blue}[1]{{\color{black}{#1}}}

\newtheorem{prop}{Proposition}
\newtheorem{defn}{Definition}
\newtheorem{coro}{Corollary}
\newtheorem{rem}{Remark}

\DeclareGraphicsExtensions{.pdf, .jpg, .jpeg, .png, .gif,.eps}
\graphicspath{{figures/}{.}}

\title{
Quantifying the overall effect of biotic interactions on species distributions along environmental gradients\\
}
\author[1,*]{Marc Ohlmann}
\author[2]{Catherine Matias}
\author[1,3]{Giovanni Poggiato}
\author[4]{St\'ephane Dray}
\author[1]{Wilfried Thuiller}
\author[4]{Vincent Miele}

\affil[1]{Univ. Grenoble Alpes, CNRS, Univ. Savoie Mont Blanc, LECA, Laboratoire d'Ecologie Alpine, F-38000 Grenoble, France}
\affil[2]{Sorbonne Universit\'e,  Universit\'e Paris Cité, Centre National de la Recherche Scientifique,  Laboratoire de Probabilit\'es, Statistique    et    Mod\'elisation,  F-75005 Paris, France}
\affil[3]{Univ. Grenoble Alpes, Inria, CNRS, Grenoble INP, LJK, F-38000 Grenoble, France}
\affil[4]{Universit\'e de Lyon, F-69000 Lyon; Universit\'e Lyon 1; CNRS, UMR5558, 
	Laboratoire de Biom\'etrie et Biologie \'Evolutive,
	F-69622 Villeurbanne, France}
\affil[*]{Corresponding author: marcohlmann@live.fr}

\date{}

\begin{document}

\maketitle


\noindent
{\bf Open Research statement:} ELGRIN is implemented in the function \texttt{elgrin} of the \texttt{R} package \texttt{econetwork} available on the code repository \url{https://plmlab.math.cnrs.fr/econetproject/econetwork} and at \texttt{CRAN} (\url{https://cran.r-project.org/}). The simulation procedure can be reproduced with the \texttt{R} code available along with this manuscript. The vertebrate data from \citet{oco20} can be found at \url{https://datadryad.org/stash/dataset/doi:10.5061/dryad.bcc2fqz79}. Climatic data were downloaded from the Worldclim v2 database (\url{http://www. worldclim.org/bioclim}) as described in the Methods, section 2.4. Land cover data were downloaded from Global Land cover v2.2 (\url{http://due.esrin.esa.int/page_globcover.php}), net primary productivity was downloaded from (\url{https://sedac.ciesin.columbia.edu/data/set/hanpp-net-primary-productivity/data-download}) and the human footprint index was downloaded from  \url{http://sedac.ciesin.columbia.edu/data/set/wildareas-v2-human-footprint-geographic}, searching for the latest version (V2 the time of the \blue{article}).\\


\pagebreak

\begin{abstract}
		
Separating environmental effects from those of interspecific interactions on species distributions has always been a central objective of community ecology. Despite years of effort in analysing patterns of species co-occurrences and the developments of sophisticated tools, we are still unable to address this major objective. A key reason is that the wealth of ecological knowledge is not sufficiently harnessed in current statistical models, notably the knowledge on interspecific interactions. 

Here, we develop ELGRIN, a statistical model that simultaneously combines knowledge on interspecific interactions (\textit{i.e.}, the metanetwork), environmental data and species occurrences  to tease apart their relative effects on species distributions. Instead of focusing on single effects of pairwise species interactions, which have little sense in complex communities, ELGRIN contrasts the overall effect of species interactions to that of the environment. 

Using various simulated and empirical data, we demonstrate the suitability of ELGRIN to address the objectives for various types of interspecific interactions like mutualism, competition and trophic interactions. \blue{We then apply the model on vertebrate trophic networks in the European Alps to map the effect of biotic interactions on species distributions.}

Data on ecological networks are everyday increasing and we believe the time is ripe to mobilize these data to better understand biodiversity patterns. ELGRIN provides this opportunity to unravel how interspecific interactions actually influence species distributions.
\end{abstract}

\pagebreak

\section*{Introduction}

Ecologists have always strived to understand the drivers of biodiversity patterns with the particular interest to tease apart the effects of environment and biotic interactions on species distributions and communities \citep{ric08,Thuiller2015,cha03,deCando_book}. Species distributions are influenced by the abiotic environment (e.g. climate or soil properties) because of their own physiological constraints that allow them or not to sustain viable populations in specific environmental configurations \citep{Austin,Pulliam}. However, the occurrence of a species in a given site is also influenced by other species through all sort of interactions that can be trophic (e.g. a predator needs preys), non-trophic (e.g. plant species need to be pollinated by insects) or competitive (two species with the same requirements might exclude each other) \citep{guisan_book, gravel2019,Lortie,Soberon}. 

Teasing apart the effects of environmental variations and interspecific interactions on species distributions and communities from observed co-occurrence patterns has always been a hot topic in ecology since the earlier debate between \citet{dia75} and \citet{con79}, to the recent syntheses on the subject \citep{bla20}. More than anything, with a few exceptions, and despite recent advances like joint species distribution models \citep{ova17} or  null model developments \citep{per01,cha13}, the conclusion has been that it is almost impossible to retrieve and estimate interspecific interactions from observed spatial patterns of species communities \citep{zur2018, bla20}. This conclusion should thus preclude any attempt to disentangle the relative effects of environment and interspecific interactions. A major difficulty of this long-standing issue is that interspecific interactions could be of any type (i.e. positive, negative, asymmetric) and that observed patterns average out all these interactions. Observed communities indeed reflect the overall outcome of interspecific interactions that is difficult to dissect, especially when analysing pairwise species spatial associations as it is commonly done \citep[e.g.,][]{Tikho17} . Yet, this overall outcome might be worth analysing on its own, for instance to measure the overall strength of interspecific interactions in a given community and between communities, how it depends on the co-existing species, and how it varies in space. 

Interestingly, so far there have been few attempts to integrate the wealth of existing knowledge to address this fundamental ecological issue \citep{bla20,hol20}. Indeed, the spatial analysis of biotic interactions is gaining an increased interest with novel technologies to measure interactions in the field (e.g. camera-traps, gut-content), open databases (e.g. GLOBI, Mangal) and the developments of new statistical tools to analyse them \citep{tyl17,pel18,ohl19,bot22}. The combination of expert knowledge, literature, available databases, and phylogenetic hypotheses has also given rise to large metanetworks that generalise the regional species-pool of community ecology by incorporating the potential interactions between species from different trophic levels along with their functional and phylogenetic characteristics \citep{Maio2020,mor15}. Despite a few attempts  \citep[e.g.,][]{sta17}, information on interaction networks has been poorly integrated to understand and model biodiversity patterns. We believe that the time is ripe to incorporate network information into the process of modelling species distributions and communities. It implies to integrate both biotic and abiotic information (and their spatial variations) as explanatory factors in statistical models to weight their relative strength. 

In this \blue{article}, we propose a novel statistical model, called ELGRIN (in reference to Charles Elton and Joseph Grinnell) that can handle the effects of both environmental factors and  known interspecific interactions (aka a metanetwork) on species distributions. We rely on Markov random fields \citep[MRF, also called Gibbs distribution, e.g.,][]{bre99}, a family of flexible models that can handle dependencies between variables using a graph.  
More specifically, ELGRIN jointly models the presence and absence of all species in a given area in function of environmental covariates and the topological structure of the known metanetwork (Figure \ref{fig:scheme} left). It separates the interspecific interaction effects (Figure \ref{fig:scheme} top-right) from those of the environment  (Figure \ref{fig:scheme} bottom-right) on species distributions. To our knowledge, ELGRIN is the first model whose outputs are the relative strengths of biotic factors needed on top of abiotic environmental variables to shape the species distributions and their spatial variation (see Latitude/Longitude in Figure \ref{fig:scheme} top-right). It thus provides a convenient way to integrate network ecology in joint species community modelling. 

In this \blue{article}, we first present the overall modelling framework and then assess its performances under different scenarios implying data simulated using three different dynamic models. In other words, although ELGRIN  considers only static observational data (metaweb and community data), we evaluated the model using simulated data generated using different dynamic models that involve various underlying processes, including intraspecific competition.
We test the ability of ELGRIN to decipher the relative importance of abiotic and interspecific interactions in these difficult cases so as to better understand what kind of signal ELGRIN can or cannot retrieve from the data. Finally, we apply the model on vertebrate trophic networks in the European Alps as an empirical study.


\section*{Material and methods}

\subsection*{Species data and potential interactions}
We consider a set of sites or locations indexed by $l \in \{1,\dots, L\}$, where the occurrence (presence/absence)  of $N$ species and a set of environmental variables (vector $W_l$) are observed. 

For the same set of $N$ species, we assume that we know all the pairwise interactions between them (e.g. who eats whom), an information summarised with a graph $G^\star=(V^\star, E^\star)$ over  the set  of nodes  $V^\star=\{1, \dots , N\}$ and edges $E^\star$. This graph, usually called a metanetwork that represents a regional pool of both species and interactions, can be obtained, for instance, by aggregating local networks at different locations or from expert knowledge and literature review  \citep[e.g.,][]{Cirtwill,Maio2020}. Note that various types of interactions can be considered here (e.g., trophic, mutualism, competition). However, while considering a mixture of interaction types is technically possible, the interpretation of results would be difficult because in our framework, $G^\star$ records the presence of an interaction and not its type. An additional note is that our model, like most species community models (e.g. Joint species distribution models, ordination techniques) relying on occurrence data, makes some assumptions about the ecological processes structuring species assemblages. In our current implementation of ELGRIN, we consider that only unimodal responses of species to environmental gradients and interspecific interactions shape communities, ignoring other processes such as dispersal limitation or mass effect for instance. 
Lastly, note also that our model supposes that the graph associated to the metanework is undirected with no self-loops (see model specifications below) and thus ignores intraspecific interactions.
Hereafter, we refer to co-present (or co-absent) species, pairs of species that are connected in the metanetwork and jointly present (or absent, respectively) at a given location.

\subsection*{The statistical model of ELGRIN}
\paragraph{Model description} 
\blue{The aim of ELGRIN model is to factorise the joint species presence distribution between a Grinnellian part, that consists in a regression on environmental covariables, and an Eltonian part that quantifies association strengths between species distribution according to the metanetwork. More formally,} we consider a set of random variables $\{\Xil\}_{i \in  V^\star}$ taking  values in  $\{0, 1\}$ and that represent the presence/absence of species $i\in V^\star$ at location $l\in \{1,\ldots,L\}$.
We rely on a {\it Markov random field} \citep[see for instance][]{bre99} to model the dependencies between species occurrences at location $l$. \blue{This is a multivariate  model that encodes statistical dependencies between species distribution using a network}. \blue{In our ELGRIN model,} these dependencies are encoded through the metanetwork $G^\star$. For each location $l\in \{1,\ldots,L\}$, we thus assume that these random variables are distributed according to a Gibbs distribution specifying the joint associations between the species occurrence variables $\{\Xil\}_{i \in V^\star}$, as follows: 
\begin{subequations}
\begin{align}
\pr(\{\Xil\}_{i \in V^\star}) =\frac 1 Z \exp\Big( &\sum_{i  \in V^\star}  [a_l +a_i + W_l^\intercal b_i + (W_l^2)^\intercal c_i ] \Xil 
\label{eq:GRF1}
\\
&+ \betap \sum_{(i,j) \in E^\star} \1\{\Xjl =\Xil=1\}
\label{eq:GRF2}
\\
&+ \betaa \sum_{(i,j) \in E^\star}  \1\{\Xjl =\Xil=0\}
\Big),
\label{eq:GRF3}
\end{align}
\label{eq:GRF}
\end{subequations}
where $\1\{A\}$ is the indicator function of event $A$ (either co-absence $\Xjl =\Xil=0$ or co-presence $\Xjl =\Xil=1$), \blue{notation $U^\intercal$ stands for the transpose of vector $U$} and $Z$ a normalising constant discussed below.
Some  model parameters have an ecological interpretation (Table \ref{tab:param}).
The use of $W_l$ and $W_l^2$ (the vector of  coordinate-wise squared values of $W_l$)   allows modelling \blue{a quadratic} species response to environmental gradient, following \blue{then} a bell-shaped relationship as expected under classical niche theory \citep{cha03}. 

Sub-equation~\eqref{eq:GRF1} is the Grinnellian part of ELGRIN, as it represents some prior probability of species occurrences independently of their interactions. 
Parameters $a_i,b_i,c_i$ capture the response of species $i$ to environment, seen through a vector of environmental covariates $W_l$. The intercepts $a_i$ and $a_l$ are estimated up to a constant only (see Appendix S1: Section~\ref{SI_sec:ident}) and may not be interpreted,
whereas the vectors $b_i,c_i$ deal with the species  environmental niche, like in a standard species distribution model \citep{guisan_book}. 

Sub-equations~\eqref{eq:GRF2} and \eqref{eq:GRF3} form the Eltonian part of ELGRIN. It considers only interactions $(i,j)\in E^\star$, i.e. the  edges of the metanetwork. 
The $\beta_l$ represent the overall influence of the interactions (as encoded through $G^\star$) on all species presence/absence at location $l$. However, this influence may be different for co-presence and co-absence, with parameters $\betap$ and $\betaa$ respectively (see Table \ref{tab:interpret}). When a $\betap$ is positive, it represents a positive driving force of co-presence on species distributions. By contrast, a negative value  indicates that species co-presences are avoided. The same reasoning holds with $\betaa$ for co-absences. 
Since the interaction parameter $\betaa$ can also be influenced by co-absences between species that are both absent at  location $l$ only because of unsuitable environmental conditions, we introduced a compatibility matrix so that the effect of interactions is only estimated in the environmental conditions where interacting species could co-occur (details are given in Appendix S1: Section~\ref{SI_sec:compat}). \blue{Importantly, this compatibility matrix is estimated during the inference procedure and is not a required input by the user.}

Note that we chose the parameters $\beta_l$ to be specific to location $l\in \{1,\ldots,L\}$ such that the effect of species interactions can vary across space. 
Finally,  $Z$ is a normalising constant that cannot be computed for combinatorial reasons, although the statistical inference procedure takes care of it. Full details of the estimation procedure and parameter identifiability are available in Appendix S1:  Section~\ref{SI_sec:estim} and Appendix S1: Section~\ref{SI_sec:ident}, respectively.

Lastly, it is important to note two specificities of the metanetwork $G^\star$ in our modelling procedure: it cannot  be  directed nor contain self-loops. Indeed, Markov random fields specify conditional dependencies between random variables $\{\Xil\}$ in an undirected way, and self-loops have no meaning in this framework. Our model assumes that these dependencies are given by the interaction network without considering the direction of edges. Consequently, this statistical model of interaction cannot be read in the light of causality. In case of trophic interactions, it consists in assuming that presence/absence of a predator and its prey are intertwined, without specifying top-down or bottom-up control. Moreover, the absence of self-loops prevents from taking into account intraspecific effects. These effects are simply ignored by ELGRIN, as they are in any joint species distribution model or ordination technique (see Appendix S1: Section~\ref{sec_SI:LVintra}).

ELGRIN is implemented in \texttt{C++} for efficiency and is available in the function \texttt{elgrin} of the \texttt{R} package \texttt{econetwork} available on the code repository \url{https://plmlab.math.cnrs.fr/econetproject/econetwork} and at \texttt{CRAN} (\url{https://cran.r-project.org/}). \blue{We assessed the performance of the method in inferring parameters from data sampled and re-sampled under the model (see Appendix S1: Section ~\ref{SI_sec:samp_resamp}).}

\begin{table}
	\begin{tabular}{ll}
		\hline
		\hline
		Variables & Ecological interpretation \\
		\hline
		\hline
		$G^\star$ & Metanetwork of known interactions (undirected)\\
  		\blue{$V^\star$} & \blue{Species (node set) of the metanetwork}  \\
    	\blue{$E^\star$} & \blue{Interactions (edge set) of the metanetwork}\\
		$\Xil$ & Presence/absence of species $i$ at location $l$\\
		$W_l$ & Environmental covariates at location $l$\\
		\hline
		\hline
		Parameters &  \\
		\hline
		\hline
		$a_i$ & Species $i$ intercept\\
		$a_l$ & Location $l$ intercept \\		
		$b_i, c_i$ & Environmental (abiotic) parameters of species $i$ \\
		$\betap$ & Co-presence strength (or avoidance when $<0$) at location $l$\\
		$\betaa$ & Co-absence strength (or avoidance when $<0$) at location $l$\\
		\hline
		\hline
	\end{tabular}
	\caption{Definition of variables and parameters of the Markov random field model ELGRIN.}
	\label{tab:param}
\end{table}


\paragraph{Model interpretation} 
In the hypothetical example where $G^\star$ is an empty graph (no edges, none of the species interact), the random variables $\{\Xil\}_{i \in V^\star}$ are independent and each species is present with probability $e^{\alpha_{i,l}} /(1+e^{\alpha_{i,l}})\in (0,1)$, where $\alpha_{i,l}=a_{l}+a_i+ W_l^\intercal b_i+(W_l^2)^\intercal c_i$. In other words, $\alpha_{i,l}$ is the logit of the probability of presence of species $i$ at location $l$ in the absence of interactions. Assuming that we have included all important environmental covariates, that there is no other ecological processes involved, and no model mis-specifications, $\alpha_{i,l}$ is analogous to the fundamental niche parameters of the species \cite[sensu][]{hut59}. It gives the probability of presence of species $i$ at location $l$ when only environmental filtering occurs.

In the case of species interactions, $G^\star$ is a non empty graph and the presence/absence information is smoothed across neighbouring nodes in $G^\star$. In Table~\ref{tab:interpret}, we detailed the ways both $\betap$ and $\betaa$ parameters capture how the metanetwork influences species co-occurrences in a given location,  notably the co-presence or co-absence of pairs of interacting species. \blue{This table describes  expected patterns of species distribution according to the combination of positive, negative and zero values for the $\beta$ parameters.} More precisely, when  species are known to interact positively (e.g. $G^\star$ encodes mutualism) and that these interactions, averaged over all species with suitable environmental conditions at location $l$,  influence their co-occurrences at that location, $\betap$ and/or $\betaa$ will be estimated as positive. On the other hand, in case of negative interactions (e.g. $G^\star$ encodes competition) that influence the co-occurrences at location $l$ of species with favorable environmental conditions, the parameters $\betap$ and/or $\betaa$ will be negative, co-presence configurations (or co-absence, respectively) tend to be avoided, meaning that only one of the two species tends to be present. 
 \blue{Given a location with fixed total number of interacting co-present (resp. interacting co-absent) species, the larger the absolute value of $\betap$ (resp. $\betaa$),  the stronger the strength of the interactions.}

\begin{figure}
	\begin{center}
		\tikzstyle{isabsent}=[circle, thick, draw=black, fill=white!25,minimum size=12pt, inner sep=0pt]
		\tikzstyle{ispresent} = [isabsent, fill=black!24]
		\tikzstyle{edge} = [draw,thick,-]
		\begin{tikzpicture}
			\tikzset{
				sh2n/.style={shift={(0,1)}},
				sh2s/.style={shift={(0,-1)}},
				sh2e/.style={shift={(1,0)}},
				sh2w/.style={shift={(-1,0)}},
				sh2nw/.style={shift={(-1,1)}},
				sh2ne/.style={shift={(1,1)}},
				sh2sw/.style={shift={(-1,-1)}},
				sh2se/.style={shift={(1,-1)}},
				rc/.style={rounded corners=2mm,line width=2pt},
				place/.style={draw,circle,fill=cyan!10,inner sep=.5mm,minimum size=5mm},
			}
			\node (metanetwork) at (0,8) [shape=rectangle,label = (a)] {
                \begin{tikzpicture}[scale=0.9, auto,swap]
					\foreach \pos/\name in {{(0,2)/2},{(2,2)/3},{(0,0)/6},{(1.5,0)/7},{(2.5,0)/8},{(1,3)/1},{(0,1)/4},{(2,1)/5}}
					\node[ispresent] (\name) at \pos {$\name$};
					\foreach \source/ \dest in {1/2,1/3,2/4,3/4,3/5,4/6,5/6,5/7,5/8}
					\path[edge] (\source) -- (\dest);
				\end{tikzpicture}
			};
			\node (presence) at (0,4) [shape=rectangle,label = (b)] {
				 \begin{tabular}{l||*{3}{c}}\hline
					\backslashbox{Sites}{Species}
					&\makebox[1em]{1}&\makebox[2em]{2}&\makebox[2em]{...}\\\hline\hline
					1 &1&1&0\\\hline
					2 &1&0&0\\\hline
					... &0&0&1\\\hline
				\end{tabular}
			};
			\node (enviro) at (0,0) [shape=rectangle,label = (c)] {
				\begin{tabular}{l||*{3}{c}}\hline
					\backslashbox{Site}{Environment}
					&\makebox[1em]{Var1}&\makebox[2em]{Var2}&\makebox[2em]{...}\\\hline\hline
					1 & 1.2 & 1.5 & 0.9\\\hline
					2 & 5.5 & 5.2 & 3.1\\\hline
					... & 0.1 & 0.2 & 0.3 \\\hline
				\end{tabular}
			};
			\node[draw] at (10,6) (map) [shape=rectangle,label = (d)] {
				\includegraphics[width=4cm]{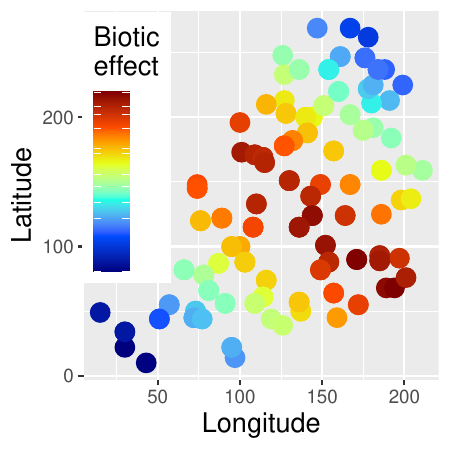}
			};
			\node[draw] at (10,1)  (niche)  [shape=rectangle,label = (e)] {
				\includegraphics[width=4cm]{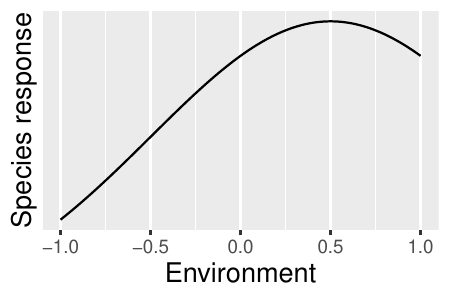}
			};
			\node[draw,font=\fontsize{15}{0}\selectfont] at (6,3.5)  (ELGRIN)     {ELGRIN};
			\draw[-stealth,rc] (metanetwork) -- (ELGRIN);
			\draw[-stealth,rc] (enviro) -- (ELGRIN);
			\draw[-stealth,rc] (presence) -- (ELGRIN);
			\draw[-stealth,rc] (ELGRIN) -- (map);height=5.2cm
			\draw[-stealth,rc] (ELGRIN) -- (niche);
		\end{tikzpicture}
           \caption{Schematic view of ELGRIN \blue{statistical} model. Given \blue{(a) an interaction } metanetwork \blue{that summarises known interactions (edges) between species (nodes),} (b) species occurrences data and (c) environmental covariates for a set of sites, ELGRIN model estimates (d) the overall effect of known biotic interactions on species distributions in each site using two association parameters, and (e) the environmental response of each species along all sites using regression parameters on environmental covariates.}
	\label{fig:scheme}
	\end{center}
\end{figure}

\tikzstyle{isabsent}=[circle, thick, draw=black, fill=white!25,minimum size=12pt, inner sep=0pt]
\tikzstyle{ispresent} = [isabsent, fill=black!24]
\tikzstyle{edge} = [draw,thick,-]
\begin{table}
 	\begin{tabular}{c||c|c|c}
	& $\betap\ll 0$ & $\betap=0$ & $\betap\gg 0$\\
	& (avoided co-presence) & (random presence) & (favored co-presence)\\ 
	\hline
	\hline    
	\pbox{20cm}{$\betaa\ll 0$ \\ (avoided co-absence)} & 
    \begin{tikzpicture}[scale=0.7, auto,swap]
    \foreach \pos/\name in {{(0,2)/2},{(2,2)/3},{(0,0)/6},{(1.5,0)/7},{(2.5,0)/8}}
        \node[ispresent] (\name) at \pos {$\name$};
    \foreach \pos/\name in {{(1,3)/1},{(0,1)/4},{(2,1)/5}}
        \node[isabsent] (\name) at \pos {$\name$};
    \foreach \source/ \dest in {1/2,1/3,2/4,3/4,3/5,4/6,5/6,5/7,5/8}
        \path[edge] (\source) -- (\dest);
    \end{tikzpicture} 
    & 
    \begin{tikzpicture}[scale=0.7, auto,swap]
    \foreach \pos/\name in {{(1,3)/1},{(0,1)/4},{(2,1)/5},{(0,0)/6}}
        \node[ispresent] (\name) at \pos {$\name$};
    \foreach \pos/\name in {{(0,2)/2},{(2,2)/3},{(1.5,0)/7},{(2.5,0)/8}}
        \node[isabsent] (\name) at \pos {$\name$};
    \foreach \source/ \dest in {1/2,1/3,2/4,3/4,3/5,4/6,5/6,5/7,5/8}
        \path[edge] (\source) -- (\dest);
    \end{tikzpicture} 
    &
    \begin{tikzpicture}[scale=0.7, auto,swap]
    \foreach \pos/\name in {{(1,3)/1},{(2,2)/3},{(0,1)/4},{(2,1)/5},{(1.5,0)/7}}
        \node[ispresent] (\name) at \pos {$\name$};
    \foreach \pos/\name in {{(0,2)/2},{(0,0)/6},{(2.5,0)/8}}
        \node[isabsent] (\name) at \pos {$\name$};
    \foreach \source/ \dest in {1/2,1/3,2/4,3/4,3/5,4/6,5/6,5/7,5/8}
        \path[edge] (\source) -- (\dest);
    \end{tikzpicture} 
	\\ 
	\hline
	\pbox{20cm}{$\betaa=0$ \\ (random absence)} & 
    \begin{tikzpicture}[scale=0.7, auto,swap]
    \foreach \pos/\name in {{(0,2)/2},{(2,2)/3},{(1.5,0)/7},{(2.5,0)/8}}
        \node[ispresent] (\name) at \pos {$\name$};
    \foreach \pos/\name in {{(1,3)/1},{(0,1)/4},{(2,1)/5},{(0,0)/6}}
        \node[isabsent] (\name) at \pos {$\name$};
    \foreach \source/ \dest in {1/2,1/3,2/4,3/4,3/5,4/6,5/6,5/7,5/8}
        \path[edge] (\source) -- (\dest);
    \end{tikzpicture} 
    & 
    \begin{tikzpicture}[scale=0.7, auto,swap]
    \foreach \pos/\name in {{(0,2)/2},{(2,1)/5},{(0,0)/6},{(1.5,0)/7}}
        \node[ispresent] (\name) at \pos {$\name$};
    \foreach \pos/\name in {{(1,3)/1},{(2,2)/3},{(0,1)/4},{(2.5,0)/8}}
        \node[isabsent] (\name) at \pos {$\name$};
    \foreach \source/ \dest in {1/2,1/3,2/4,3/4,3/5,4/6,5/6,5/7,5/8}
        \path[edge] (\source) -- (\dest);
    \end{tikzpicture} 
    &
    \begin{tikzpicture}[scale=0.7, auto,swap]
    \foreach \pos/\name in {{(0,1)/4},{(2,1)/5},{(0,0)/6},{(1.5,0)/7}}
        \node[ispresent] (\name) at \pos {$\name$};
    \foreach \pos/\name in {{(1,3)/1},{(0,2)/2},{(2,2)/3},{(2.5,0)/8}}
        \node[isabsent] (\name) at \pos {$\name$};
    \foreach \source/ \dest in {1/2,1/3,2/4,3/4,3/5,4/6,5/6,5/7,5/8}
        \path[edge] (\source) -- (\dest);
    \end{tikzpicture} 
    \\
	\hline
	\pbox{20cm}{$\betaa\gg 0$ \\ (favored co-absence)} & 
    \begin{tikzpicture}[scale=0.7, auto,swap]
    \foreach \pos/\name in {{(0,2)/2},{(0,0)/6},{(1.5,0)/7},{(2.5,0)/8}}
        \node[ispresent] (\name) at \pos {$\name$};
    \foreach \pos/\name in {{(1,3)/1},{(2,2)/3},{(0,1)/4},{(2,1)/5}}
        \node[isabsent] (\name) at \pos {$\name$};
    \foreach \source/ \dest in {1/2,1/3,2/4,3/4,3/5,4/6,5/6,5/7,5/8}
        \path[edge] (\source) -- (\dest);
    \end{tikzpicture} 
    & 
    \begin{tikzpicture}[scale=0.7, auto,swap]
    \foreach \pos/\name in {{(1,3)/1},{(0,2)/2},{(1.5,0)/7},{(2.5,0)/8}}
        \node[ispresent] (\name) at \pos {$\name$};
    \foreach \pos/\name in {{(2,2)/3},{(0,1)/4},{(2,1)/5},{(0,0)/6}}
        \node[isabsent] (\name) at \pos {$\name$};
    \foreach \source/ \dest in {1/2,1/3,2/4,3/4,3/5,4/6,5/6,5/7,5/8}
        \path[edge] (\source) -- (\dest);
    \end{tikzpicture} 
    & 
    \begin{tikzpicture}[scale=0.7, auto,swap]
    \foreach \pos/\name in {{(1,3)/1},{(0,2)/2},{(2,2)/3},{(0,1)/4}}
        \node[ispresent] (\name) at \pos {$\name$};
    \foreach \pos/\name in {{(2,1)/5},{(0,0)/6},{(1.5,0)/7},{(2.5,0)/8}}
        \node[isabsent] (\name) at \pos {$\name$};
    \foreach \source/ \dest in {1/2,1/3,2/4,3/4,3/5,4/6,5/6,5/7,5/8}
        \path[edge] (\source) -- (\dest);
    \end{tikzpicture} 
    \\
	\end{tabular}
	\caption{Simplified view of the different behaviours of the model in function of the parameters $\betap$ and $\betaa$. The graph represents the metanetwork containing all potential interactions where species can be either present (gray node) or absent (white node) in a given location $l$ leading to different estimated $\betap$ and $\betaa$.
	 When $\betap \ll 0$ or $\betaa \ll 0$, interacting species in the metanetwork tend to avoid each other: whenever one is absent, the other tend to be present and reversely. This situation favors a checkerboard pattern on the metanetwork. Reversely, whenever $\betap\gg 0$ (resp. $\betaa\gg 0$), there are groups of interacting species that tend to be all present (resp. all absent), inducing sets of gray (resp. white) neighbour nodes in the metanetwork. Whenever $\betap=0$ or $\betaa=0$, there are sets of interacting species whose states are independent from one another and thus purely random (the proportions of gray and white nodes are governed by the values of the parameters in the Grinnellian part of the model). 
	}
	\label{tab:interpret}
\end{table}
\subsection*{Exploration on simulated data from complex dynamic processes}
To test the ability of ELGRIN to infer the overall biotic and abiotic controls on species distributions, we used three theoretical models, different from the one underlying ELGRIN, to dynamically simulate spatial community data with $50$ species and $400$ sites along a single environmental gradient and  combined them with multiple different interactions scenarios (competition, mutualism, and no interaction). To do that, we chose species niche optima evenly distributed along a single environmental gradient. The metanetworks were built so that interacting species have close niche optima (otherwise they would never co-occur). In the mutualistic scenario, we also considered a case where species that facilitate each other tend to have an abiotic niche that is also not too close (otherwise they would compete). 
Along this single environmental gradient, niche optima and associated metanetworks according the interaction scenarios, we used three theoretical dynamic models (Lotka-Volterra, colonisation-extinction, and  co-existence model aka VirtualCom) to simulate the resulting species distribution data. These models have different underlying assumptions and processes, which allowed testing ELGRIN under a total of 9 different configurations. 

\paragraph{Lotka-Volterra model}
The Lotka-Volterra model is one of the foundational models in community ecology \citep{Takeuchi}. 
This model simulates communities under both intra- and interspecific interactions, while ELGRIN is not able to handle intraspecific interactions (its metanetwork does not allow for self-loops). 
Thus we parameterized the Lotka-Volterra simulation with intraspecific interactions being negligible in regards to interspecific interactions. That way we generated species community data that meets the type of data and ecological questions ELGRIN is designed to tackle 
(for details, see Appendix S1: Section~\ref{SI_sec:LV}). Nonetheless, we also explored the converse case to fully understand  the limits of ELGRIN (see Appendix S1: Section~\ref{sec_SI:LVintra}). 

\begin{figure}[ht!]
	\begin{center}
		\includegraphics[height=8cm]{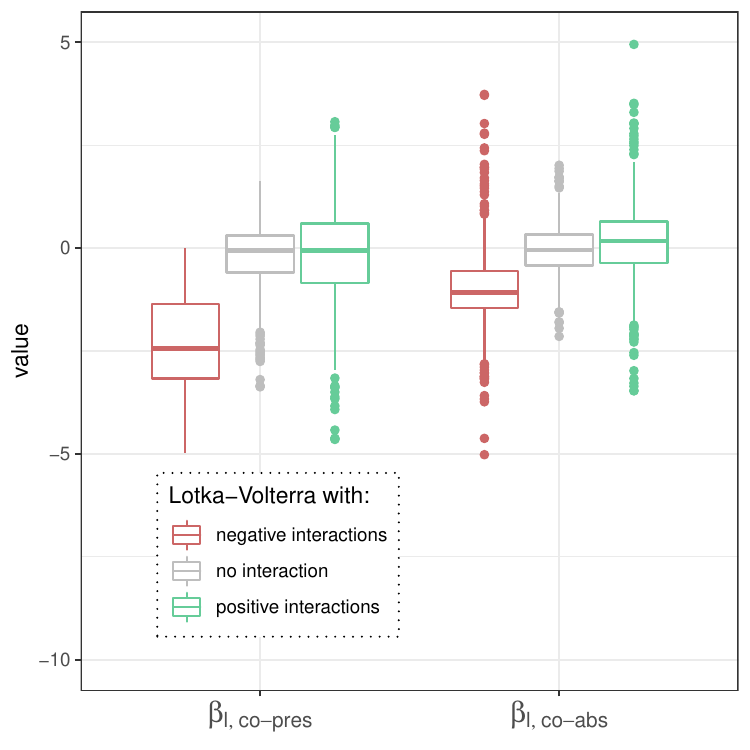}
		\caption{Distribution of co-presence ($\betap$) and co-absence  ($\betaa$) strengths inferred using ELGRIN on simulated ecological communities using a Lotka-Volterra model with competition (negative interactions),  mutualism (positive interactions) or no interactions.}
		\label{fig:simu_LV}
	\end{center}
\end{figure}

\paragraph{Colonisation-extinction model}
We used an updated version of the stochastic colonisation-extinction model developed in \citet{ohlmann2022assessing} to simulate the species community dataset for the three interaction scenarios (for details see Appendix S1: Section~\ref{SI_sec:CE}). The model consists in a multivariate Markov chain that converges towards a stationary distribution from which we sampled the species community dataset. 

\begin{figure}[ht!]
	\begin{center}
		\includegraphics[height=8cm]{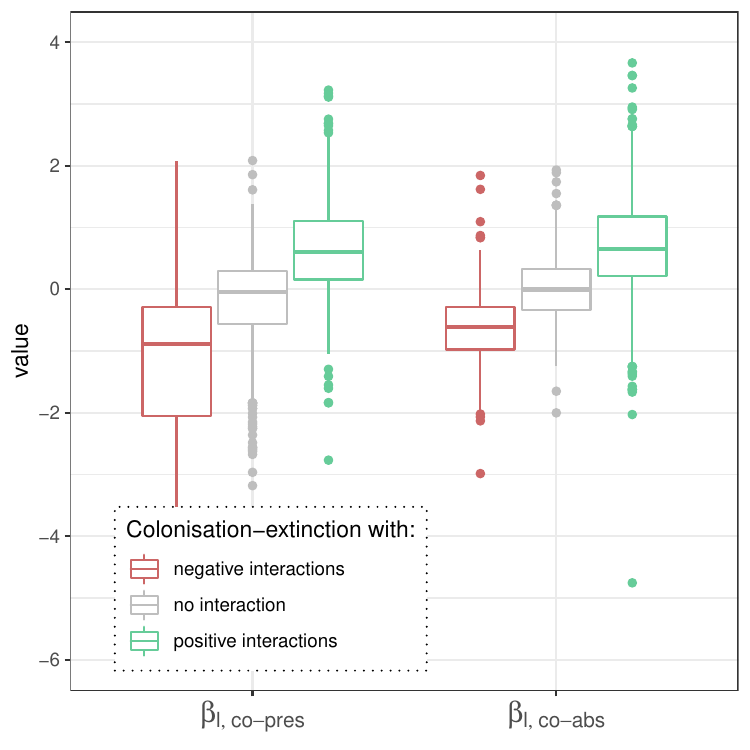}
		\caption{Distribution of co-presence ($\betap$) and co-absence  ($\betaa$) strengths inferred using ELGRIN  on simulated ecological communities using a colonisation-extinction model with competition (negative interactions), mutualism (positive interactions) or no interactions.}
		\label{fig:simu_CE}
	\end{center}
\end{figure}

\paragraph{VirtualCom model}
We used an updated version of the model developed by \citet{VirtualComTM} to simulate communities whose composition is driven simultaneously by biotic and abiotic environmental effects, for the three interaction scenarios  (for details see Appendix S1: Section~\ref{SI_sec:VC}). In this model, each community has the same carrying capacity (i.e. the exact number of individuals in each location). 

\begin{figure}[ht!]
	\begin{center}
		\includegraphics[height=8cm]{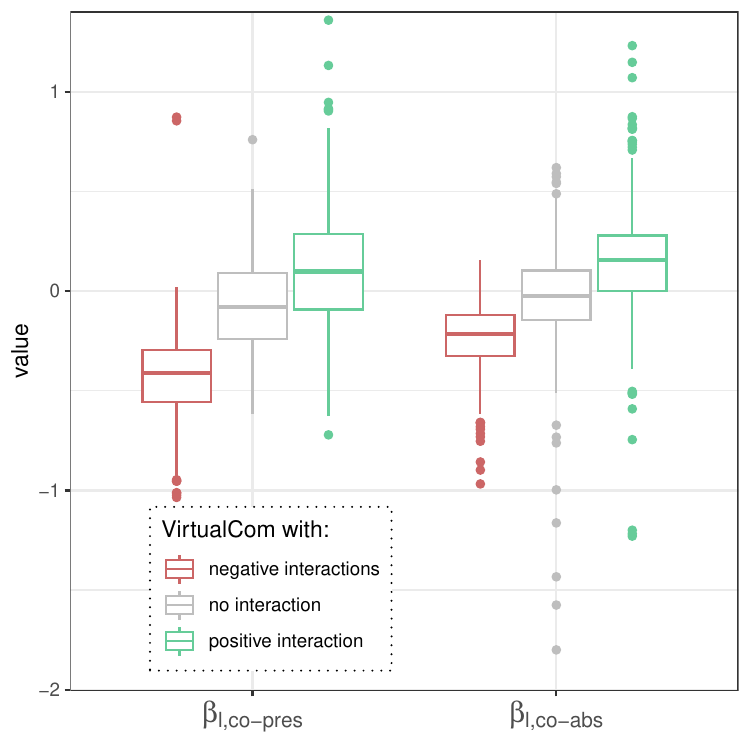}
		\caption{Distribution of co-presence ($\betap$) and co-absence  ($\betaa$) strengths inferred using ELGRIN  on simulated ecological communities with VirtualCom model, with competition (negative interactions),  mutualism (positive interactions) or no interactions.}
	\label{fig:simu_VC}
	\end{center}
\end{figure}

\subsection*{Application: a case study}
We analyse the newly available Tetra-EU 1.0 database, a species-level trophic network of European tetrapods \citep{Maio2020} that combines all known potential interactions between terrestrial mammals, birds, reptiles and amphibians occurring in Europe. This metanetwork is based on data extracted from known interactions, scientific literature, including published \blue{articles}, books, and grey literature \citep[see][for a complete description of the data and the reference list used to build the metanetwork]{Maio2020}. As usual with such data, this metanetwork does not provide information on interaction plasticity or intraspecific interactions. 
We restricted our analyses on the European Alps that show sharp environmental gradients and varying trophic web distributions \citep{oco20}. We extracted the species distribution data from \citet{mai13} at a 300 m resolution. We upscaled all species ranges maps to a 10x10 km equal-size area grid and cropped the distribution data to the European Alps. Species were considered present on a given 10x10 km cell if they were present in at least one of the 300 x 300 m cells within it. This yielded species distributions maps for 257 breeding birds, 99 mammals, 36 reptiles, and 30 amphibians over 2138 locations. 
Environmental covariates were extracted at the same resolution and were selected following previous work on those data \citep{bra19}. For climate, we used mean annual temperature, temperature seasonality, temperature annual range, total annual precipitation and coefficient of variation of precipitation that were all extracted from the Worldclim v2 database (\url{http://www.worldclim.org/bioclim}). Using GlobCover (GlobCover V2.2; \url{http://due.esrin.esa.int/page\_globcover.php}), we extracted the number of habitats present in a given pixel, habitat diversity in a given pixel based on Simpson index and habitat evenness as a measure of habitat complexity. Finally, we added an index of annual net primary productivity (Global Patterns in Net Primary Productivity, v1 (1995), \url{http://sedac.ciesin.columbia.edu/data/set/hanpp-net-primary-productivity}) and the human footprint index (\url{http://sedac.ciesin.columbia.edu/data/set/wildareas-v2-human-footprint-geographic}). Since these data were highly correlated, we used a PCA to retain the three leading vectors as environmental covariates ($W_l$) in ELGRIN.

\section*{Results}

\subsection*{Tests on simulated species community data}
\blue{Let us first recall that we assessed the performance of the method in inferring parameters from data sampled and re-sampled under ELGRIN model (see Appendix S1: Section ~\ref{SI_sec:samp_resamp}). We now turn to more involved dynamical theoretical models.}

For the three theoretical models (Lotka-Volterra, colonisation-extinction and VirtualCom), ELGRIN was correct in identifying the no interaction scenario, with estimated interaction strengths close to 0 (Figures \ref{fig:simu_LV}, \ref{fig:simu_CE} and \ref{fig:simu_VC}). Similarly, ELGRIN was able to retrieve the negative effects of interactions in the case of competition as simulated by the three theoretical models. The   $\betap$ and $\betaa$ parameters were mostly negative  (with much higher absolute values for $\betap$), capturing the backbone of the  competitive interactions. They indicated that co-presence and co-absence were avoided (as presented in Table \ref{tab:interpret} top-left), leading to some level of competitive exclusion. In the VirtualCom co-existence model, this phenomenon was clearly the by-product of the competitive interactions and the carrying capacity in terms of number of individuals (that explicitly induced exclusion). 
When positive interactions come into play (i.e. mutualism), the results should be contrasted between those obtained for the Lotka-Volterra model, where ELGRIN does not \blue{qualitatively} identify the processes at stake and the two other models  (colonisation-extinction and VirtualCom) where ELGRIN succeeds in identifying them. 
The Lotka-Volterra simulation with positive interactions scenario produced species that are essentially distributed along their respective niches (see Appendix S1: Figure \ref{fig:spp_distrib_LV}). As a consequence, this distribution can be simply fitted with the Grinellian part of the model and ELGRIN estimates the $\beta$s close to zero (Figure \ref{fig:simu_LV}). That means that the same dataset could have been produced by only abiotic environmental conditions and the actual species distribution does not contain anymore a pattern that ELGRIN would identify as the trace of the positive interspecific interactions. 
On the contrary, in the positive interactions scenario, with both competition-colonisation and VirtualCom co-existence models, ELGRIN correctly identified the process at play. The parameters $\betap$ and $\betaa$ were mostly positive. During the simulation steps, the presence of one species was then favored by the presence of another species it interacted with, leading to a co-presence phenomena captured by the positive $\betap$. Conversely, the inverse mechanism emerged for co-absence, implying that the $\betaa$ tended to be positive as revealed by ELGRIN (Figures  \ref{fig:simu_CE}, \ref{fig:simu_VC}). 
\blue{To quantitatively investigate the difference between $\betap$ and $\betaa$ distributions in the three simualations, we performed Kolmogorov-Smirnov (KS) tests. For each simulation, we tested whether $\betap$ and $\betaa$ distributions were significantly different in the scenarios with interactions (either positive or negative) from the scenario without interaction. In the three simulations, the tests correctly identify significant differences between interactions and no interaction scenarios (see Appendix S1: Table \ref{tab:KStest}).}


\subsection*{Empirical case study}
When fitted to the European vertebrate dataset, ELGRIN's parameters $\betap$ and $\betaa$ were highly correlated (Pearson correlation of 0.84, see Appendix S1:  Section~\ref{SI_sec:realdata1}) suggesting that trophic interactions impact both predator/prey co-presence and co-absence.
In what follows, we therefore mainly dealt with $\betap$. 

We first observed a structured spatial pattern of the effects of interactions, with regions of negative or positive $\betap$ (bluish or reddish colors respectively in Figure \ref{fig:alpesmap}). The largest $\betap$ values were found mainly in the french Alps and in the Eastern zone.

\blue{In Figure  \ref{fig:alpes}, we present the values of different variables at each location, according to groups of estimated $\betap$ parameters, where the width of each boxplot is proportional to the number of points in each class.}
Almost all the highest $\betap$ ($>0.05$) were revealed in locations below 1600 m of altitude (Figure \ref{fig:alpes}a, \blue{$p$-value of the KS test inferior to $2.2\mathrm{e}{-16}$, details given in Appendix S1: Section~\ref{SI_sec:realdata2}}). In these regions, \blue{species} richness was generally high (Figure \ref{fig:alpes}b, \blue{$p$-value inferior to $2.2\mathrm{e}{-16}$}). In the opposite, the higher up, the more likely $\betap$ was negative (Figure \ref{fig:alpes}a). This was particularly true above 1600 m in the central Alps, where almost all the negative $\betap$ were estimated (bluish colors in Figure \ref{fig:alpesmap}). Locations with negative $\betap$ have a lower species richness (Figure \ref{fig:alpes}b). 
\blue{Interestingly, locations with low connectance have lower absolute $\betap$ values (Figure \ref{fig:alpes}c, $p$-value inferior to $2.2\mathrm{e}{-16}$) indicating a lower effect of biotic interactions compared to abiotic effects in these locations. Here, connectance is the density of the graph induced by the metanetwork at location $l$, namely its nodes are species occurring  at location $l$ and edges are those from the metanetwork between those present species.} 


\begin{figure}[ht!]
	\begin{center}
	\includegraphics[height=8cm]{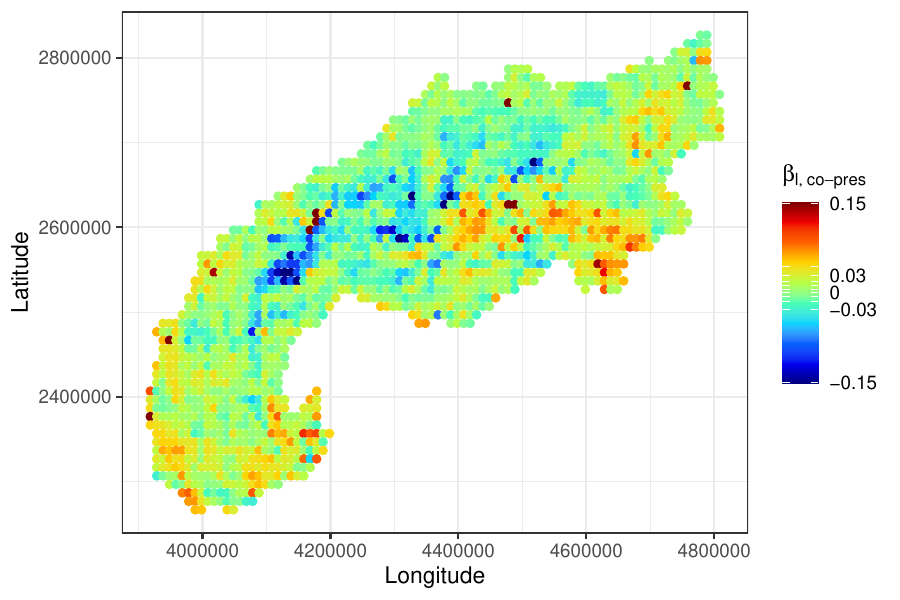}
	\caption{Results of ELGRIN on the European tetrapods case study.  Map of estimated $\betap$ (one dot per location). The color scale indicates the $\betap$ values. For the sake of representation, $\betap$ values above 0.15 in absolute value were set to 0.15.}
	\label{fig:alpesmap}
	\end{center}
\end{figure}

\begin{figure}[ht!]
	\begin{center}
		\includegraphics[height=11cm]{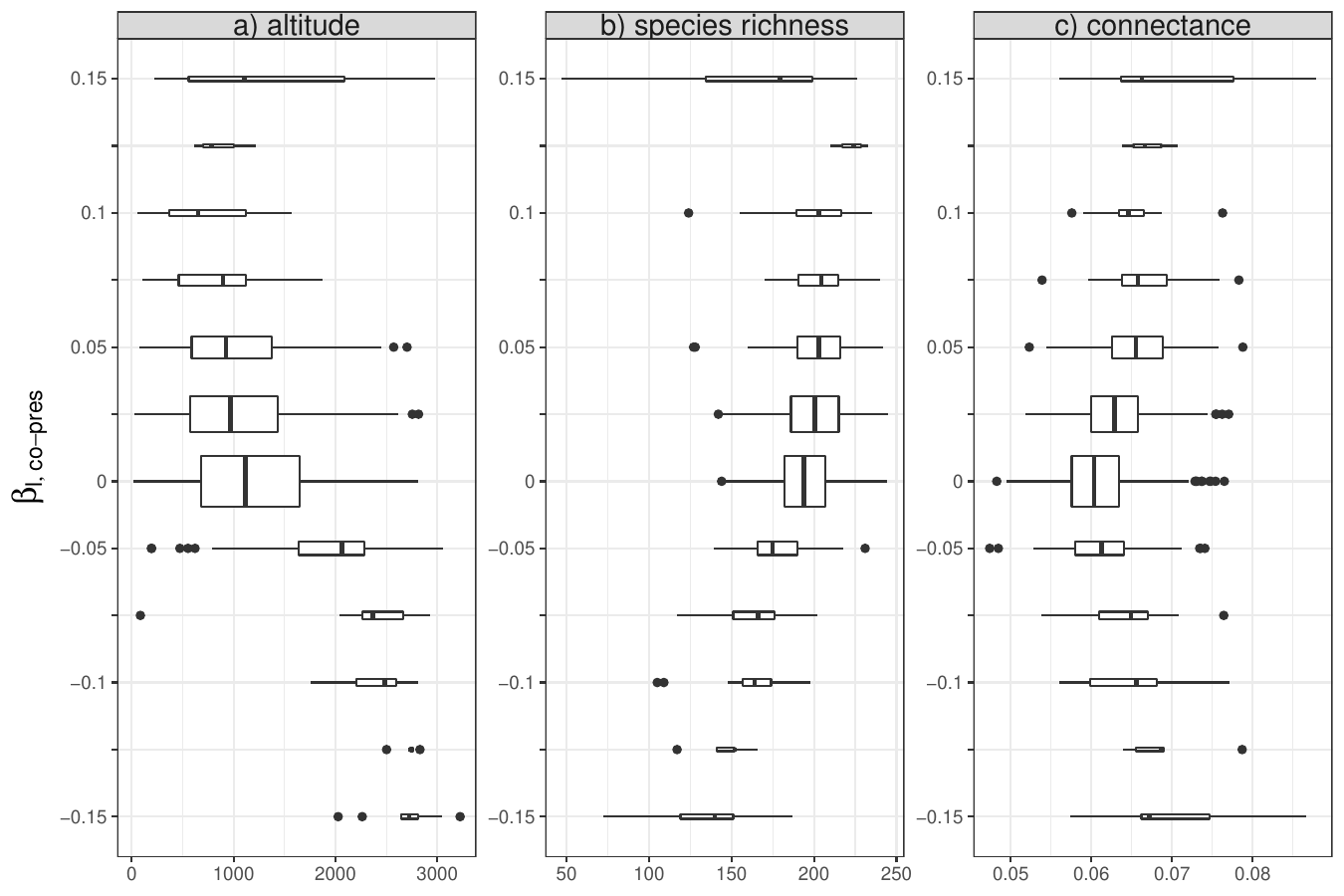}
            \caption{Results of ELGRIN on the European tetrapods case study. Boxplots representing the values of different variables at each location, according to the estimated $\betap$ values (x axis). (a) altitude, (b) species richness, 
            and (c) connectance \blue{(density of the graph induced by the metanetwork at location $l$)}
            For the sake of representation, $\betap$ values above 0.15 in absolute value were set to 0.15. \blue{Width of the boxplots is proportional to the number of points in each class.} }
	\label{fig:alpes}
	\end{center}
\end{figure}
\section*{Discussion}

Deciphering the mechanisms driving spatial patterns of species distributions and communities is likely one of the most active fields of ecological research since the early days of biogeography and community ecology. Still, there was so far no comprehensive statistical approach able to make the best of existing knowledge on interspecific interactions, species occurrence and environmental data to measure and quantify the dual effects of environment and biotic interactions on species distributions.
Our proposed model that relies on Markov random fields builds on the ability of graphical models to encode and analyse species distribution dependencies using the known species interactions.
This formalism allows, within the same model, to account for both the effects of the environment and the interspecific interactions, which reconciles the Grinnellian vision of species niches (i.e. how species respond to the abiotic environment) with its Eltonian counterpart (i.e. how species respond to the biotic environment). The mathematical foundations of ELGRIN are strong and its framework is flexible allowing for useful extensions to handle interaction strength, sampling effects and plasticity of interactions (see Appendix S1:  Section~\ref{SI_sec:gene}). 

A key element of ELGRIN is its ability to measure the overall relative effects of interspecific interactions on species distributions with respect to abiotic environmental conditions, which allows to summarise all local pairwise interactions in a single measure (i.e. $\betaa$ or $\betap$). This measure can then be mapped, related to spatial layers to understand how the overall relative effect of interspecific interactions vary in space and in function of the environment or the ecosystem types. Importantly, this measure can also be carefully investigated at a given location in function of the constituent species, trophic groups, specialists vs generalists, connectance and so on. Interestingly, we can thus see our $\beta_l$  estimates as an extended and more meaningful version of the famous checkerboard score or C-score \citep{sto90}, which has been used to quantify local interspecific interactions from co-occurrence pattern \cite[e.g., ][]{bou12}. The main advantage of ELGRIN over the C-score is that instead of trying to infer biotic interactions only from co-occurrences (which we know to be notoriously difficult, nearly impossible), it quantifies, in a conditional way,
the effects of the known interspecific interactions on species communities, while accounting for the environmental responses of the species. Our approach is thus not comparable with recent developments on joint species distribution models (JSDMs) that relate species occurrences to environmental conditions, and provides a residual covariance matrix that could be interpreted on the light of missing predictors, mis-specifications and biotic interactions \citep{ova17, zur2018}.
This matrix represents covariances between model residuals (the left-over from the environmental effects) and actually provides little information about biotic interactions \citep{zur2018, pog21}. ELGRIN does not infer any residual covariance and directly accounts for the known interactions through the metanetwork. In JSDMs, missing covariates will inevitably lead to spurious estimates of biotic interactions. In ELGRIN, the parameter $a_l$ is supposed to capture most of the unexplained information that is independent of the interspecific interactions. This parameter acts as a site random effect in mixed models and is expected to filter out the effects of missing covariates, although some remaining species-specific effects might still percolate into the $\beta_l$ estimates. \\

In the presentation of ELGRIN and in our case studies, we focused on a single interaction type at a time (e.g. competition, mutualism or trophic interaction). When dealing with a single type of interaction, competition for instance, the modelling is explicit since we clearly understand the effect that one species can have on another species. Although it is technically possible to manage a metanetwork composed of different types of interactions, the interpretation would become problematic. Different interaction types can have opposite effects, such as competition (a species excludes other species) and mutualism (a species facilitates other species) and, since ELGRIN captures an overall impact of these interactions on the distributions at each location, interpreting ELGRIN's results can be misleading in that case. Additionally, it is worth noting that since ELGRIN relies on a Markov random field, $G^\star$ is undirected. In other words, when the original metanetwork  encodes asymmetric interactions (e.g. predator-prey), they are then converted in undirected edges that only represent the presence of interactions (whatever their direction). It is thus critical to keep that in mind when interpreting the results of ELGRIN, and when merging different types of interactions together. The same issue happens when hoping to interpret the residual covariance matrix of JSDM through the lens of biotic interactions, since the values of the covariance matrix could reflect any type of interactions between species, that could be asymmetric or symmetric, or both. 
Note that we explicitly used a bell-shaped relationship for modelling species response to environmental gradients. While it would be possible to modify ELGRIN to incorporate any other parametric relationship, the actual version of ELGRIN would lead to erroneous conclusions whenever used on data where this assumption is not satisfied. \\

More generally, it is important to underline that ELGRIN finds the most likely scenario under a model associated to underlying assumptions. This model represents up to date the most reasonable and simple model that integrates both interspecific interactions and abiotic factors in modelling the species distribution. In that sense, it goes beyond (joint) species distribution models or ordination models by including explicitly the effect of interspecific interactions.
However, the most likely scenario under this model is not necessarily the real one that lead to observed data. For instance, ELGRIN was not able to identify the positive interspecific interactions present in the dynamics of a Lotka-Volterra model (even when restricting to negligible intraspecific interactions). Despite being a most widely studied model, the Lotka-Volterra model still raises important challenges. Indeed, whether the system reaches a single globally stable equilibrium point is known only in specific cases \citep{Takeuchi}. Since ELGRIN infers model interspecific interactions relative effects from the species distributions, existence of multiple equilibria in the Lotka-Voltera dynamics (depending on the initial conditions that are unknown) could pose serious identifiability problems. Even in presence of a unique and globally stable equilibrium point, several parameters or different interaction types could lead to the same equilibrium and thus same observed species distributions. This also raises tough identifiability issues. We hope that the recent developments around Lotka-Volterra model will help to circumvent those issues \citep{biroli2018marginally,remien2021structural}. 
We could easily simulate species distributions, using models that include other ecological processes, on which ELGRIN would fail in recovering the true underlying generation processes. 
Indeed we present simulations scenarios beyond the assumptions of the model (i.e., a Lotka-Volterra model with intraspecific interactions stronger than interspecific ones, see Appendix S1: Section~\ref{sec_SI:LVintra}), where ELGRIN again uncovered a completely different explanation of the data at hand. If the data contain the signature of different ecological processes (including ones not considered by ELGRIN),  ELGRIN will not be able to infer properly the relative effects of interspecific interactions and abiotic factors. The question of knowing which ecological processes could indeed be recovered from species distribution patterns remains thus debated \citep[e.g.][]{bla20}. A last note is that ELGRIN only deals with binary occurrence data rather than abundance or frequency data. In our simulation design, both the Lotka-Volterra and the VirtualCom models produced abundance data that we had to sample to obtain binary signals, loosing information during the process. On the contrary, ELGRIN performs better on colonisation-extinction simulations, where the dynamics directly generates binary data. Extending ELGRIN from the binary setup to the continuous one could improve the inference by considering more information in the species distribution data but it remains an important methodological challenge.


In terms of further perspectives, we might wonder whether this model could be extended for prediction purposes. In principle, it is possible to draw presence/absence data from the model for different values of the environment variables. These different values could allow for predictions in space but also in time.  
However, something to keep in mind is that metanetwork will not change in the model and will thus be considered as static and thus representative in space (or in time). If the metanetwork has not been built with that prediction perspective in mind, this might be an issue as we will miss interaction rewiring effects on species distributions. Instead, if the metanetwork is truly a potential metanetwork that tries to incorporate these potential interactions that have been observed yet \citep[i.e. ][]{Maio2020}, it might be interesting to investigate how biotic interactions might further influence future species distributions in response to environmental changes.


\section*{Acknowledgements}
The authors would like to thanks the anonymous reviewers that carefully reviewed a previous version of our model and manuscript, and notably proposed the Lotka-Volterra simulation scheme.
Funding was provided by the French National Center for Scientific Research (CNRS) and the French National Research Agency (ANR) grant ANR-18-CE02-0010-01 EcoNet. VM would like to thank the LECA laboratory for hosting him in Chambéry. 

\section*{Conflict of Interest Statement}
The authors declare no conflict of interest.

\newpage
\appendix

\renewcommand\thesection{S.\arabic{section}}
\renewcommand\theequation{S.\arabic{equation}}
\renewcommand\thefigure{S.\arabic{figure}}

\begin{center}
  {\Large 
Appendix S1 for manuscript 'Quantifying the overall effect of biotic interactions  on species distributions along environmental gradients', by M. Ohlmann, C. Matias, G. Poggiato, S. Dray, W. Thuiller \& V. Miele.}
\end{center}

\section{ Model extensions} 
\label{SI_sec:gene}
\subsection{Interaction strength}
Besides the binary case, it is also possible to handle interaction strengths. An interaction strength can represent a frequency (e.g., the number of visits of a pollinator to a plant), an intensity \citep[e.g., rate of predation, ][]{ber04} or a preference (e.g. modulating trophic links with known affinities of a predator to its preys). 

We write  $A^\star=(A_{ij}^\star)_{i,j\in  V^\star}$ the adjacency matrix of the  graph $G^\star$. Now, each edge $(i,j)\in E^\star$ is modulated through the weight $A_{ij}^\star$ of the interaction. 
In this case, sub-equations~\eqref{eq:GRF2} and \eqref{eq:GRF3} are replaced by  
\begin{align*}
 \betap \sum_{(i,j) \in E^\star} A_{ij}^\star\1\{\Xjl =\Xil=1\} &=\betap \sum_{(i,j) \in E^\star} A_{ij}^\star \Xjl \Xil \\
\text{and }
 \betaa \sum_{(i,j) \in E^\star} A_{ij}^\star \1\{\Xjl =\Xil=0\} &=\betaa \sum_{(i,j) \in E^\star} A_{ij}^\star (1 -\Xjl) (1-\Xil), 
\end{align*}
respectively.  

\subsection{Sampling effects} The random variables $\Xil$ that indicate the presence of species $i$ at location $l$ might not be exactly observed due to sampling effects. Here, we propose to account for these effects by assuming that each species $i\in V^\star$ is sampled with probability  $p_{i,l}\in (0,1)$ at location  $l\in \{1,\dots, L\}$.   We therefore introduce a new set
of random  variables $\Yil, i \in  V^\star, l \in \{1,\dots, L\}$ such
that each $\Yil$ only depends on $\Xil$ and is distributed as 
\begin{align*}
	\pr(\Yil | \Xil) &= p_{i,l}^{\Yil}(1-p_{i,l})^{1-\Yil} \Xil + (1-\Xil)(1-\Yil) \nonumber \\
	&= p_{i,l}^{\Xil \Yil}(1-p_{i,l})^{\Xil(1-\Yil)}  \1\{ (1-\Xil) \Yil \neq 1\}. 
\end{align*}
Specifically,  whenever  $\Xil=0$ (species  $i$  is absent  from location $l$), species $i$ cannot be observed  at  location $l$  and $\Yil=0$.  Now,  when  $\Xil=1$ (species   $i$  is  present  at  location $l$),  it  is observed ($\Yil=1$) with sampling probability $p_{i,l}$ and unobserved ($\Yil=0$) with probability $1-p_{i,l}$. The parameter $p_{i,l}$ must be given by the user considering three possible cases: species dependent sampling ($p_{i,l}:=  p_i ; i \in V^\star$), location dependent sampling ($p_{i,l} := p_l ; 1\le l \le L$) or constant sampling ($p_{i,l} := p$). 
In this case, the $\Xil$ become latent variables as we only observe the $\Yil$'s. The model turns out to be a hidden Markov random field (HMRF). 

\subsection{Plasticity of interactions} Our model is able to assume that interactions are not necessarily induced by the presence/absence variables (we can assume that two species interact in a given location but not in another location). In this case, we consider a sample of observed graphs $G^1,\dots , G^L$  where each $G^l=(V^l, E^l)$ is such that $V^l  \subset V^\star$. These graphs represent local interactions that are observed at the different locations $l \in \{1, \dots, L\}$.  The main point here is that we assume that these interactions are sampled from the pool of potential interactions encoded in the metanetwork $G^\star$. 
Let $\Al =(\Aijl)_{i,j\in V^l}$ denote the adjacency matrix of the graph $G^l$. 
We assume that any two species that are observed and that can potentially interact (i.e., are linked in the metanetwork $G^\star$) do effectively interact at location $l$ with a probability that depends only on these two species. Namely for any $(i,j) \in E^\star$, conditional on
the fact  that two species  $i,j\in V^\star$  were observed at  location $l$ (namely $\Yil\Yjl=1$), we set 
\begin{equation*}
	\Aijl | \Yil\Yjl=1 \sim \mathcal{B}(\epsilon_{ij}),   
\end{equation*}
and $\Aijl \equiv 0$ whenever $(i,j)\notin E^\star$ or $\Yil=0$ or $\Yjl=0$. 
This additional parameter $\epsilon=\{\epsilon_{i,j}\}_{i,j \in V^\star}$ allows us to handle interaction plasticity directly in the model.


\section{ Mathematical details on the model}

\subsection{Identifying the parameters of the Gibbs distribution}
\label{SI_sec:ident}
We first address the issue of the identifiability of the parameters from the Gibbs distribution. 
In what follows, we focus on the case of a binary metanetwork $G^\star$. However, our results remain valid in the weighted case, where degrees are replaced by weighted degrees and the cardinality $|E^\star|$ (total number of edges in $G^\star$) becomes the total sum of the weights. 

Let us focus on the model with no covariates ($W_l=0$) and consider for each location $l\in \{1,\dots,L\}$ the maps $\psi_l=(\{a_{i}\}_i, a_{l}, \betap,\betaa) \mapsto \prpsil$, where 
\begin{align*}
\prpsil(\{\Xil\}_{i \in V^\star}) =\frac 1 {Z_{\psi_l}} \exp\Big( &\sum_{i  \in V^\star} (a_{i}+a_l)  \Xil 
+ \betap \sum_{(i,j) \in E^\star} \Xjl \Xil
\\
&+ \betaa \sum_{(i,j) \in E^\star} (1-\Xjl)(1- \Xil)
\Big).
\end{align*}
For any $\psi =(\{a_{i}\}_{i,l}, \{a_{l}, \betap,\betaa\}_l)$ we also define the global probability distribution $\prpsi$ as follows
\[
\prpsi(\{\Xil\}_{i \in V^\star ; 1\le l \le L}) = \prod_{l=1}^L \prpsil(\{\Xil\}_{i \in V^\star}). 
\]

\begin{prop}[Identifying linear combinations of the parameter]
\label{prop:ident}
In the model without covariate ($W_l=0$, for any $l$), the probability distribution $\prpsi$ uniquely  defines the quantities
\begin{align}
&    \betap+\betaa , \label{eq:beta_sum}\\
\text{and } &    a_{i}+a_{l} +\betap \degi 
\text{ or equivalently }   a_{i}+a_{l}  -\betaa \degi ,
\label{eq:alpha_beta} 
\end{align}
for any  $i\in V^\star, l\in \{1,\dots, L\}$, 
where $\degi$ is the degree of species $i$ in the metanetwork $G^\star$.
Moreover, if there exist 2 species $1\le i,j\le N$ such that $\degi\neq \degj$ in $G^\star$, then the probability distribution $\prpsi$ uniquely  defines the additional quantities
\begin{align}
    &\betaa-\betaalp \label{eq:beta_diff_loc} 
\text{ or equivalently }   \betap-\betaplp, \\
\text{and }  & a_l - a_{l'}, \label{eq:al}
\end{align}
for any $l, l' \in \{1,\dots, L\}$. 
\end{prop}

\begin{proof}
Let us denote $\alpha_{i,l}= a_i+a_l$. As $\prpsil$ is a marginal of $\prpsi$, we start by fixing the location $l\in \{1,\dots, L\}$ and consider the probabilities of specific configurations at this location. 
We let $\Xmil$ denote the set $\{\Xjl ; j\in V^\star, j\neq i \}$. From the knowledge of $\prpsi$, we obtain  for $l \in \{1,...,L\}$ and $i \in V^{\star}$ the quantities 
\begin{align*}
s_0^l :=& \log \prpsil(\{0,...,0\}) = -\log(Z_{\psi_l}) + |E^\star|\betaa
\\
s_1^l :=&  \log \prpsil(\{1,...,1\}) = -\log(Z_{\psi_l}) + \sum_i \alpha_{i,l} +|E^\star|\betap
\\
s_{10}^{i,l} :=& \log \prpsil(\{\Xil=1,\Xmil=0\}) = -\log(Z_{\psi_l}) + \alpha_{i,l} + \betaa(|E^\star| - \degi)\\
s_{01}^{i,l} :=& \log \prpsil(\{\Xil=0,\Xmil=1\}) = -\log(Z_{\psi_l}) + \sum_{j \neq i} \alpha_{j,l} + \betap(|E^\star| - \degi) , 
\end{align*}
where 
$|E^\star|$ is the cardinality of the set $E^\star$. 
It follows
\begin{align*}
r_{1}^l &:= s_1^l - s_0^l
 = \sum_i \alpha_{i,l} +|E^\star|(\betap-\betaa)\\
r_2^{i,l} &:= s_{10}^{i,l}- s_0^l 
= \alpha_{i,l} - \betaa\degi
\\
r_{3}^{i,l}  &:= s_{01}^{i,l} - s_0^l
= \sum_{j\neq i} \alpha_{j,l} + (\betap-\betaa) |E^\star| - \betap \degi .
\end{align*}
From these equations, 
we uniquely obtain 
\begin{align*}
    t_1^{i,l} :=&  r_1^l -r_3^{i,l}= \alpha_{il} +\betap \degi \\
    t_2^{i,l}  :=& r_1^l - r_2^{i,l} -r_3^{i,l} = (\betaa+\betap) \degi .
\end{align*}
As a consequence, as soon as there is at least one edge in the metanetwork $G^\star$ (inducing at least one species $i$ with $\degi\neq 0$) we can obtain the quantities $\betaa+\betap$ (recall that $\degi$ is known) as well as $\alpha_{i,l} +\betap \degi$  uniquely from the distribution $\prpsi$. Note also that combining the knowledge of these two quantities, the second is equivalent to knowing  $\alpha_{i,l} -\betaa \degi$. \\

Now, let us recall that $\alpha_{i,l}=a_i+a_l$. For two different locations $l\neq l'$, we have access to 
\[
t_1^{i,l} -t_1^{i,l'} = a_l - a_{l'} + (\betap-\betaplp)\degi. 
\]
We now assume that there exist two species  $1\le i,j\le N$ such that $\degi\neq \degj$ in $G^\star$ and obtain \eqref{eq:beta_diff_loc} as follows 
\[
\betap-\betaplp = (t_1^{i,l} -t_1^{i,l'}- t_1^{j,l} +t_1^{j,l'} ) [\degi-\degj]^{-1}. 
\]
Combining this with \eqref{eq:beta_sum}, it is equivalent to the unique identification of $\betaa-\betaalp$. Finally, going back to $t_1^{i,l} -t_1^{i,l'} $ we uniquely obtain $a_l-a_{l'}$. 
\end{proof}


\begin{defn}[Equivalence class] \label{defn:equiv}
For any parameter $\psi =(\{a_{i}\}_{i}, \{a_{l}, \betap,\betaa\}_l)$, its  equivalence class $[\psi]$ is defined as 
\[
[\psi] :=\{ (\{a_{i}+\gamma \degi- \delta \}_{i}, \{a_{l} +\delta, \betap-\gamma,\betaa+\gamma\}_l) ; \gamma \in \R, \delta \in \R \}.
\]
\end{defn}

\begin{coro}[Parameter identifiability up to the equivalence class]
\label{coro:ident}
In the model without covariate ($W_l=0$, for any $l$) and assuming that  there exist 2 species $1\le i,j\le N$ such that $\degi\neq \degj$ in $G^\star$, we have that  whenever there are two parameter values $\psi, \tilde \psi$ such that $\prpsi =\mathbb{P}_{\tilde{\psi}}$, then $\tilde \psi \in [\psi]$. In other words, the equality  $\prpsi =\mathbb{P}_{\tilde{\psi}}$ implies that there exist real values $\gamma, \delta  \in \R$ such that for any $i\in V^\star$ and $l\in \{1,\dots,L\}$, we have
\begin{align*}
\tilde a_i &=    a_{i}+ \gamma \degi +\delta \\
\tilde a_l &= a_{l} -\delta\\
\tilde {\beta}_{l,co-pres}&= \betap-\gamma\\
\tilde {\beta}_{l,co-abs}&= \betaa+\gamma . 
\end{align*}
\end{coro}

\begin{proof}
Assume that  $\prpsi =\mathbb{P}_{\tilde{\psi}}$ and define for any location $l\in \{1,\dots,L\} $ the quantity 
$\gamma_l := \betap-\tilde{\beta}_{l,co-pres}$. 
We know from Proposition~\ref{prop:ident} that 
\begin{align*}
 \betaa+\betap &= \tilde{\beta}_{l,co-abs}+\tilde{\beta}_{l,co-pres} \\
\tilde{a}_i+\tilde{a}_l+\tilde{\beta}_{l,co-pres}\degi& =a_i+a_l+\betap\degi .
\end{align*}
This induces that 
\begin{align*}
\gamma_l &= \tilde{\beta}_{l,co-abs}-\betaa \\    
\text{and}\quad \tilde{a}_i+\tilde{a}_l  &=a_i+a_l +\gamma_l\degi.
\end{align*}
 Let us further prove that $\gamma_l$ does not depend on $l$. 
 From Proposition~\ref{prop:ident}  and the additional assumption that at least two species have different degrees in the metanetwork, we have for any locations $l,l'\in \{1,\dots, L\}$, 
\[
    \betap-\betaplp = \tilde{\beta}_{l,co-pres} - \tilde{\beta}_{l',co-pres} = \betap-\betaplp -\gamma_l +\gamma_{l'}, 
\]
which implies that $\gamma_l=\gamma_{l'}$ for any pair of locations. 
Finally, let us define for any location and any species 
\[
\delta_l =a_l - \tilde{a}_l  \quad \text{and} \quad 
\delta_i =a_i - \tilde{a}_i .
\]
We have established that $\delta_l+\delta_i = - \gamma \degi$. This implies that $\delta_l$ is constant through locations and equal to some $\delta$. This concludes the proof. 
\end{proof}

Corollary~\ref{coro:ident} tells us that the model parameter is identifiable up to the equivalence class in Definition~\ref{defn:equiv}. 
Note that it is possible to choose one specific  representative parameter in this class. 

\begin{prop}[Choosing a representative] 
\label{prop:param_choice}
In the model without covariate ($W_l=0$, for any $l$) and assuming that  there exist 2 species $1\le i,j\le N$ such that $\degi\neq \degj$ in $G^\star$, for any parameter value $\tilde \psi$, it is possible to choose a unique representative $\psi \in [\tilde \psi]$ such that the estimated linear regression coefficients of the set of parameters $\{a_i\}_i$ over the degrees $\{\degi\}_i$ are equal to 0, namely 
\[
(\hat \gamma, \hat \delta) := \argmin_{(\gamma, \delta) \in \R^2} \sum_{i\in V^\star} (a_i -\gamma\degi -\delta) ^2 
\]
satisfies $(\hat \gamma, \hat \delta)=(0,0)$.
\end{prop}

\begin{proof}
Fix a parameter value $\tilde \psi$ and consider the linear regression  of the set of parameters $\{\tilde a_i\}_i$ over the degrees $\{\degi\}_i$, namely 
\[
(\tilde \gamma, \tilde \delta) := \argmin_{(\gamma, \delta) \in \R^2} \sum_{i\in V^\star} (\tilde a_i -\gamma\degi -\delta) ^2 .
\]
Then by setting the parameter $\psi=(\{a_{i}\}_{i,l}, \{a_{l}, \betap,\betaa\}_l)$ as 
\begin{align*}
a_i &:= \tilde a_i -\tilde \gamma\degi -\tilde \delta ; \\
a_l & := \tilde a_l + \tilde \delta ;\\
\betap &:= \tilde{\beta}_{l,co-pres} +\tilde \gamma \\
\betaa &:= \tilde{\beta}_{l,co-abs} -\tilde \gamma 
\end{align*}
(for any $i,l$), we know from Definition~\ref{defn:equiv} that $\psi\in [\tilde\psi]$ and also by definition, the estimated values 
\[
(\hat \gamma, \hat \delta) := \argmin_{(\gamma, \delta) \in \R^2} \sum_{i\in V^\star} (a_i -\gamma\degi -\delta) ^2
\]
will now satisfy $(\hat \gamma, \hat \delta)=(0,0)$.
\end{proof}

\begin{rem}
The choice of the representative parameter given by Proposition~\ref{prop:param_choice} is such that the response of species $i$ to the environment does not depend on its degree in the metanetwork and thus on its number of interactions. This is a natural choice to separate the Grinellian part from the Eltonian one in our model. 
Note that this representative parameter is the one we rely on when interpreting the model. Thus, when we comment the behaviour of the model with respect to different values of its parameter, we always rely on this specific representative.  

Note however that whatever the choice of the representative, the intercept values $a_i$ and $a_l$ are inferred up to an additive constant.

\end{rem}

\subsection{A compatibility matrix to robustify the model}
\label{SI_sec:compat}
In this section, we slightly modify the model to handle cases where either there are species with tight environmental niches or where the metanetwork $G^\star$ contains edges between species with incompatible environmental niches (which would be a nonsense). Indeed, we aim at estimating Eltonian effects only when species are in their Grinnelian niche. 

We introduce a binary matrix  $C=(C_{il})_{i \in V^\star, 1\le l\le L}$ that encodes the possibility for species $i$ to be present  at location $l$ given its niche properties. 
The matrix $C$ is called a \emph{compatibility matrix}. \blue{For the model's presentation,} it is supposed to be fixed and known.
\blue{In practice, it is either obtained from expert knowledge, otherwise built from the realized niche of each species (our implementation in the function \texttt{elgrin} will pre-estimate the compatibility matrix from realized niche before fitting the model)}. In the latter case, for any species $i$, at each location $l$ and for each covariate $d$, relying on the observation set $\{\xil\}_{i,l}$, we set 
\begin{align}
\label{SI_eq:interval}
    \omega_{id} &= \inf_{1\le l \le L} \{W_{ld} ; \xil=1 \} ,\\
    \Omega_{id} &= \sup_{1\le l \le L} \{W_{ld} ; \xil=1 \} \\
 \text{and }   C_{il} &= 1\{ \forall 1\le d\le D, W_{ld} \in [\omega_{id}; \Omega_{id}]\}. \notag
\end{align}
where location $l$ is characterized by an environmental covariate vector $W_l=(W_{l1},\dots,W_{lD})$. Naturally,  if $\Xil=1$ then $C_{il}=1$.

Relying on the compatibility matrix, at each location $l$ we restrict our attention to species compatible with the environment at this location. In particular, we now impose that $\Xil=0$ whenever $C_{il}=0$. 
Thus the probability distribution of the species in ELGRIN is modified as follows
\begin{align*}
\prpsil(\{\Xil\}_{i \in V^\star}) 
& = \left(\prod_{i\in V^\star ; C_{il}=0}(1-\Xil)\right) \times  \frac 1 {Z_{\psi_l}} \exp\Big\{ \sum_{i  \in V^\star ; C_{il}=1}   \Big[(a_l +a_i + W_l^\intercal b_i + (W_l^2)^\intercal c_i)  \Xil  \\
& + \betap \sum_{j; (i,j) \in E^\star} \Xjl \Xil  
+ \betaa \sum_{j ; (i,j) \in E^\star} C_{jl} (1-\Xjl)(1-\Xil) \Big] \Big\}.
\end{align*}
Note that if the compatibility matrix is full of 1 (i.e. all the species may occur at all locations), we are back to our initial model. Otherwise, we now avoid mistaking co-absence of two interacting species with the event of two independent absences due to incompatible niches. 

From a modeling point of view, the modified version of the model helps in robustifying our results. This is the case for instance when considering interacting species with tight niches. Indeed, at locations $l$ where two interacting species $i,j$ are absent due to incompatible environmental conditions (i.e. $C_{il}=C_{jl}=0$), we observe that $\Xil=\Xjl=0$. In that case in our  original model, this double absence would wrongly be interpreted as a co-absence and blur the inference of $\betaa$. 
Note also that whenever two species $i,j$ are potentially interacting (i.e. $(i,j)\in E^\star$), we consider that their respective niches should overlap  ($C_{il}=C_{jl}=1$ for at least one location $l$). 
If this rule is not satisfied, it could happen that, 
without the additional factor  $C_{il}C_{jl}$ regulating the co-absence term, an absence of species $i$ would be interpreted as a co-absence due to its interaction with species $j$. 

Note that at locations $l$ where the environment covariates $W_l$ prevent from the occurrence of a species $i$ (i.e. $C_{il}=0)$, it is useless to try to fit the Grinellian part of the model, i.e. the non-informative intercepts $a_i,a_l$ and the parameters $b_i, c_i$. So that when appropriate, we only consider the estimated maps $W \mapsto  W^\intercal b_i + (W^2)^\intercal c_i$ on the environment values compatible with species $i$.

\subsection{Hidden Markov random field and its interpretation}
We discuss here the model in its full generality, including possible weights on the metanetwork, sampling effects, plasticity of interactions and the robust version relying on a  compatibility matrix. 
We thus have $\bX:= \{\bX^l\}_{1\le l \le L}=\{\Xil\}_{i  \in V^\star, 1\le l \le L}$ (resp.  $\bY:= \{\bY^l\}_{1\le l \le L}=\{\Yil\}_{i \in
	V^\star, 1\le l \le L}$ and $\A:=\{\Al\}_{1\le l \le L}=\{\Aijl\}_{i,j \in V^l, 1\le l \le L}$) denoting the set of true occurrence variables (resp. observed occurrences and observed interactions). 
We assume that we observe $(\bY,\A)$, while $\bX$ are latent random variables. 

A Gibbs distribution specifies the joint associations between the species occurrence variables $\{\Xil\}_{i \in V^\star}$, as follows
\begin{align}
\prpsil(\{\Xil\}_{i \in V^\star}) =&\left(\prod_{i\in V^\star ; C_{il}=0}(1-\Xil)\right) \times  \frac 1 {Z_{\psi_l}} \exp\Big\{ \sum_{i  \in V^\star ; C_{il}=1}   \Big[(a_l +a_i + W_l^\intercal b_i + (W_l^2)^\intercal c_i)  \Xil  \nonumber\\
& + \betap \sum_{j; (i,j) \in E^\star} \Xjl \Xil  
+ \betaa \sum_{j ; (i,j) \in E^\star} C_{jl} (1-\Xjl)(1-\Xil) \Big] \Big\}.
\label{eq:GRF_compat}
\end{align}

First note that the normalizing constant  $Z_{\psi_l}$  is given by
\begin{align*}
Z_{\psi_l}= &  \sum_{i    \in    V^\star ; C_{il}=1} \sum_{x_i   \in    \{0,1\} }
\exp\Big( \sum_{i  \in V^\star ; C_{il}=1} [a_{l}+a_i+ W_l^\intercal b_i+(W_l^2)^\intercal c_i]  x_i \\
&+ \betap\sum_{j ; (i,j)
	\in E^\star} A_{ij}^\star x_i x_j + \betaa\sum_{j ; (i,j) \in E^\star} A_{ij}^\star C_{jl} (1-x_i) (1-x_j)  \Big).
\end{align*}
In general, this  normalising constant $Z_{\psi_l}$  cannot be computed due to the large  number of possible configurations appearing in the sum. The statistical inference procedure needs to deal with that.

The model interpretation strongly builds on the {\it Markov property}, a fundamental characteristic of Markov random fields. In the following we focus on the species compatible with one location ($C_{il}=1$); otherwise recall that its occurrence is set to zero with probability 1. Let us denote $\Ni$ the set of species  $j\in V^\star$ that are connected to $i$ in the graph $G^\star$ (namely $\{j \in V^\star ;  A^\star_{ij}\neq 0\}$) and $\XNvl$, the set of corresponding random  variables $\Xjl$ for $j\in \Ni$. We also recall that  $\Xmil$ denotes the set $\{\Xjl ; j\in V^\star, j\neq i \}$. Then, under the {\it Markov property} we have 
\begin{align}
\prpsil   (\Xil   |   \Xmil, C_{il}=1)   =  \prpsil   (\Xil   |   \XNvl, C_{il}=1)   \propto
\exp\Big( & [a_l + a_i+ W_l^\intercal b_i+(W_l^2)^\intercal c_i] \Xil  \nonumber\\
& + \betap \sum_{j  \in \Ni} A_{ij}^\star \Xjl\Xil \nonumber\\
 &+ \betaa \sum_{j  \in \Ni} A_{ij}^\star  C_{jl}(1-\Xjl)(1-\Xil)\Big),
\label{eq:Markov}
\end{align}
where $\propto$ means proportional (equals up to a normalising constant). 
More specifically, it means that the conditional occurrence probability of a species $i$ is modulated by the occurrences of the species interacting with $i$ in $G^\star$. In other words, a species presence only depends on abiotic environment and on the species it interacts with. Moreover, the presence/absence variables of any two species are not statistically independent of each other if $G^\star$ is connected (namely, if there exists a path between any two species in $G^\star$). Meanwhile, if $G^\star$ has more than one connected component \citep[i.e. disconnected compartments,][]{kra03}, then the presence/absence of species in different components are independent. 
The Markov property is the cornerstone idea of our model. Indeed, the conditional probabilities of each random variable is specified through~\eqref{eq:Markov} and is rooted on the idea that the occurrence of a species $i$ at location $l$ depends both on a suitability term, specific to that species and the local environment, and on the presence/absence of other species with whom it interacts (as encoded in the metanetwork). From this set of conditional probabilities, the Hammersley-Clifford theorem \citep{Besag74} ensures that there exists a proper joint distribution on the random variables $\{X_i^l\}_{i,l}$ and that it is given by Equation~\eqref{eq:GRF_compat}. \\


Now,  the observed species occurrence variables $\Yil, i \in  V^\star, l \in \{1,\dots, L\}$ are distributed such
that each $\Yil$ only depends on $\Xil$ (the true occurrence variable) with 
\begin{align}
\pr(\Yil | \Xil) &= p_{i,l}^{\Yil}(1-p_{i,l})^{1-\Yil} \Xil + (1-\Xil)(1-\Yil) \nonumber \\
&= p_{i,l}^{\Xil \Yil}(1-p_{i,l})^{\Xil(1-\Yil)}  \1\{ (1-\Xil) \Yil \neq 1\}. 
\label{eq:obs_species}
\end{align}
In what follows, we choose to impose that the sampling parameters $p_{i,l}$ are set by the user. A consequence of this is that the quantity~\eqref{eq:obs_species} will play no role in the inference procedure. Indeed, it is a constant quantity with respect to the parameter. 
Finally we set 
\begin{equation}
\label{eq:obs_interac}
\Aijl | \Yil\Yjl=1 \sim \mathcal{B}(\epsilon_{ij}),   
\end{equation}
and $\Aijl \equiv 0$ whenever $(i,j)\notin E^\star$ or $\Yil=0$ or $\Yjl=0$. 

Building on Equations~\eqref{eq:obs_species} and~\eqref{eq:obs_interac}, we first obtain 
the  conditional  distribution  of all  observations
$(\bY , \A)$ given the latent variables $\bX$ 
\begin{align*} 
& \prphi(\bY, \A | \bX) = \prod_{l=1} ^L   \prphi(\Al | \bY^l)  \pr(\bY^l | \bX^l) \\
&= \prod_{l=1} ^L \prod_{i    \in    V^\star}
\Big[  p_{i,l}^{\Xil \Yil}(1-p_{i,l})^{\Xil(1-\Yil)}  1\{ (1-\Xil) \Yil \neq 1\} \Big]
\times \prod_{(i,j)\in E^\star}
\epsilon_{ij}^{\Yil \Yjl \Aijl} (1-\epsilon_{ij})^{\Yil \Yjl (1-\Aijl) } . 
\end{align*}
Here, the  parameter $\epsilon=\{\epsilon_{ij}\}_{i,j \in	V^\star}$
	drives the distribution of the observation  process  from
the latent  one. 

Finally,  our model is obtained by combining this with Equation~\eqref{eq:GRF_compat} for the distribution of the latent
variables $\bX$. Thus   the  global model   is  parameterised   by  $\theta=\{\theta_l\}_{1\le l \le L}$ where each $\theta_l
=(\psi_l,\epsilon)$.  This amounts to the following sets of parameters
$$(\{a_i,b_i,c_i\}_{i\in V^\star},\{a_l,\betaa, \betap\}_{1\le l \le L},   
\{\epsilon_{ij}\}_{i,j \in V^\star})$$ 
so there are $3N + 3L+N(N-1)$ parameters when the observed graphs $A^l$ are directed (and $3N + 3L+N(N-1)/2$ when the observed graphs $A^l$ are undirected) compared with $N(N-1)L$ observations. However note that in the model inference (see next section),  the parameters $\epsilon_{ij}$ are pre-estimated  (see Equation~\eqref{eq:epsilon_chapo}) and do not appear in the main inference algorithm (see Algorithm~\ref{algo:SEM}). 
In what follows, we often use the notation
\[
\alpha_{i,l}= a_i + a_l +  W_l^\intercal b_i + (W_l^2)^\intercal c_i .
\]

A chain graph \citep{Lauritzen_book} describing the dependencies among
the random variables in this model is given in
Fig.~\ref{fig:chain}.  

\begin{figure}[h]
	\centering
	\begin{tikzpicture}
	\node (X1) at (1,3) {$X_1$};
	\node (X2) at (2,3) {$X_2$};
	\node (X3) at (3,3) {$X_3$};
	\node (X4) at (4,3) {$X_4$};
	\node (X5) at (5,3) {$X_{5}$};

	\node (Y1) at (1,1.5) {$Y_1$};
	\node (Y2) at (2,1.5) {$Y_2$};
	\node (Y3) at (3,1.5) {$Y_3$};
	\node (Y4) at (4,1.5) {$Y_4$};
	\node (Y5) at (5,1.5) {$Y_{5}$};

	\node (A12) at (-1,0) {$A_{12}$};
	\node (A13) at (-0.2,0) {$A_{13}$};
	\node (A14) at (0.6,0) {$A_{14}$};
	\node (A15) at (1.4,0) {$A_{15}$};
	\node (A23) at (2.2,0) {$A_{23}$};
	\node (A24) at (3,0) {$A_{24}$};
	\node (A25) at (3.8,0) {$A_{25}$};
	\node (A34) at (4.6,0) {$A_{34}$};
	\node (A35) at (5.4,0) {$A_{35}$};
	\node (A45) at (6.2,0) {$A_{45}$};

	\path[-]
	(X1) edge (X2)
	(X2) edge (X3) ;
	\draw[-] (X1) to [bend left] (X3);
	\draw[-] (X1) to [bend left] (X4);
	\draw[-] (X3) to [bend left] (X5);
	
	\path[->] 
	(X1) edge (Y1)
	(X2) edge (Y2)
	(X3) edge (Y3)
	(X4) edge (Y4)
	(X5) edge (Y5) ; 
	\path[->] 
	(Y1) edge (A12)
	(Y1) edge (A13)
	(Y1) edge (A14)
	(Y1) edge (A15) 
	(Y2) edge (A12)
	(Y2) edge (A23)
	(Y2) edge (A24)
	(Y2) edge (A25)
	(Y3) edge (A13)
	(Y3) edge (A23)
	(Y3) edge (A34)
	(Y3) edge (A35)
	(Y4) edge (A14)
	(Y4) edge (A24)
	(Y4) edge (A34)
	(Y4) edge (A45)
	(Y5) edge (A15)
	(Y5) edge (A25)
	(Y5) edge (A35)
	(Y5) edge (A45);
	
	\end{tikzpicture}
	\caption{Example of a metanetwork $G^\star$ (relations among the random variables $\{X_i\}_{i \in V^\star} $  with $V^\star
		= \{1,\dots, 5\}$, on the top
		row) and induced dependency chain graph of all the variables  in the model for one observed
		undirected graph $A=(A_{ij})_{i < j }$ with no self-loops.  }
	\label{fig:chain}
\end{figure}
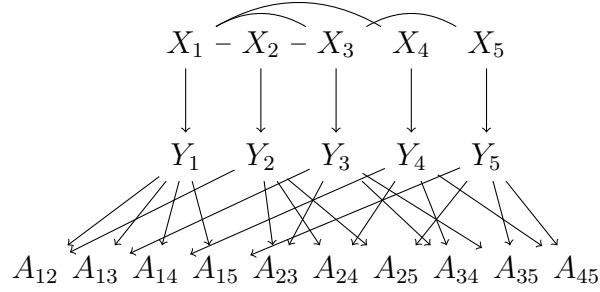

\section{ Model inference}
\label{SI_sec:estim}
We present the inference procedure in the most general case, namely with weighted metanetwork, sampling effects and plasticity of interactions. 
This means that our inference procedure takes place in the context of a hidden Markov random field model. 

\subsection{Likelihood}
The  log-likelihood for  observing independent  interaction  graphs $G^1,\dots  , G^L$ at the different locations (and thus species occurrences variables ; indeed it is equivalent to observe $G^1,\dots  , G^L$ or  $(\bY^1,A^1,\dots , \bY^L , A^L)$) in this model is given by 
\begin{equation*}
\ell_{n,L}(\theta) = \sum_{l=1}^L \log \prtl( G^l) ,
\end{equation*}
where 
\[
\prtl( G^l)  = \sum_{\{\xil\}_{i  \in V^\star}  \in \{0,1\}  ^N} \prtl(G^l,
\{\Xil= \xil ; i \in V^\star\}).
\]
As usual in latent variables models, this sum over all possible configurations $\{\xil\}_{i  \in V^\star}  \in \{0,1\}  ^N$ cannot be computed (unless $N$ is  really small).  The inference procedure in latent variable models generally relies on the Expectation-Maximisation (EM) algorithm \citep{DLR77}. In the context of hidden Markov random fields, many difficulties arise that prevent from using this simple strategy.

The complete log-likelihood $\ell_{n,L}^{c}(\theta) $
contribution  of all observations  and all  latent  configurations is given by 
\begin{align*}
\ell_{n,L}^{c}(\theta)  :=& \log \prt(\bX, G^l,\dots , G^L) = \sum_{l=1}^L \log \prtl (\bX^l, \bY^l, \Al) \\
=&  \sum_{l=1}^L \log \prpsil(\bX^l)  +\sum_{l=1}^L\sum_{i \in  V^\star}\log \pr(\Yil|\Xil  ) 
+\sum_{l=1}^L \sum_{i,j\in V^l} \log \prphi(\Aijl| \Yil,\Yjl). 
\end{align*}
This can be written as 
\begin{align*}
\ell_{n,L}^{c}(\theta) =& \sum_{l=1}^L\sum_{i
	\in V^\star}  C_{il}\log(1-\alpha_{i,l})+\sum_{l=1}^L\sum_{i
	\in V^\star} C_{il}\Xil \log \left( \frac{\alpha_{i,l}}{1-\alpha_{i,l}}\right) +  \sum_{l=1}^L  \sum_{(i,j)
	\in E^\star} C_{il} C_{jl}A_{ij}^\star \Xjl \Xil \\
&+  \sum_{l=1}^L  \betaa\sum_{(i,j) \in E^\star} A_{ij}^\star C_{il} C_{jl}(1-\Xjl)( 1-\Xil) 
- \sum_{l=1}^L  \log (Z_{\psi_l})\\
&+   \sum_{i  \in   V^\star} \sum_{l=1}^L  \Xil  \Big\{\Yil\log   (   p_{i,l})   +
(1-\Yil)\log (1-p_{i,l}) \Big\}\\
&+\sum_{i,j\in V^\star } \sum_{l=1}^L \Yil \Yjl \Big\{\Aijl \log \epsilon_{ij}
+(1-\Aijl) \log (1-\epsilon_{ij}) \Big\} +\text{cst}.
\end{align*}
Here, we  restrict our attention  to complete datasets  $(\bX^l, G^l)$
which are compatible,  in the sense that whenever  $\Xil=0$ we also
have $\Yil=0$.  Otherwise the probability  above is  0 and its  log is
$-\infty$.

\subsection{Estimating the frequency of interactions}
First, it is important to note that a consequence of the dependence among the
$\{\Xil\}_{i \in  V^\star}$ is  that the random  variables $\Aijl$ and $A_{i'j'}^l$ are dependent.
However, this dependency is entirely carried by the
species  observations  $\Yil$'s   (which  themselves  are  dependent
through the species  latent presences $\Xil$'s). In  other words, we
have  $\prphi(\Al  |  \bY^l,  \bX^l  ) =\prphi(\Al  |  \bY^l)  $.  A
consequence is that the parameters  $\epsilon$ that describe the graph
distribution  are directly  estimated  from the  data. While the
sampling parameters and the random field ones ($\betaa,\betap$ and
$\alpha_{i,l}$'s) require  a sophisticate inference procedure, the $\epsilon_{ij}$ parameters are directly estimated by the frequencies  
\begin{equation}
\hat \epsilon_{ij} = \frac{ \sum_{l=1}^L A_{ij}^l } {\sum_{l=1}^L \Yil \Yjl}.
    \label{eq:epsilon_chapo}
\end{equation}
Here, the normalising term $\sum_{l=1}^L \Yil \Yjl$ is simply the number of simultaneous observations of species $i$ and $j$ across the $L$ different locations, while the numerator counts the number of observed interactions between those species across locations.

\subsection{Inference of the random field parameters with simulated field algorithm}
Now, we focus on the estimation of random field parameters $\betaa,\betap$ and
$\alpha_{i,l}$'s. 
A  classical  \texttt{EM}  algorithm would  consist  in  (iteratively)
optimising with respect to $\psi=\{\psi_l \}_{1\le l \le L}$ the quantity 
\begin{equation}
\label{eq:EM}
Q(\psi)=  \sum_{l=1}^L \esp\big( \log  \prpsil(\bX^l, \bY^l) |
\psilt , \bY^l \big) = \sum_{l=1}^L \esp\big[ \log  \prpsil(\bX^l)  \big| \psilt , \bY^l \big] +\text{cst},  
\end{equation}
computed with the current value of the parameter $\psit=\{\psilt\}_{1\le l \le L}$. (Recall that in our setup, the observations $\bY$ are obtained from $\bX$ through a random function with fixed and known parameters). 
The above quantity has many drawbacks: first it contains the partition functions
$Z_{\psi_l}$ that are unknown and cannot be computed. Second, the conditional distribution of
$\bX^l$ given $\bY^l$ has an intricate dependency structure
and thus may not be computed (in fact it is also a Markov random field). 

We thus follow the \emph{simulated field algorithm}  proposed in
\cite{Celeux_etal}.  It is  based on  two different  approximations of
probability distributions plus a
simulation step, as follows.  First, the distribution $\prpsi(\bX)$ appearing in
the  complete likelihood  is replaced  by a  mean-field approximation,
namely the product distribution 
\begin{equation}
\label{eq:approx1}
\pr^{1} (\bX |\psi, \btx) = 
\prod_{l=1}^L \prod_{i \in V^\star} \prpsil(\Xil | \XNvl = \xtNvl),  
\end{equation}
for  some   well  chosen  fixed  configuration   $\btx=(\tilde x^l_i )_{1\le l \le L, i \in V^\star}$.  Second,  the
conditional distribution  $\prpsi(\bX |  \bY)$ used for  integrating the
complete      log-likelihood     in~\eqref{eq:EM} is also replaced
by a mean-field approximation, that is 
\begin{equation}
\label{eq:approx2}
\pr^{2} (\bX |\psi, \btx,\bY) = 
\prod_{l=1}^L \prod_{i \in V^\star} \prpsil(\Xil | \XNvl = \xtNvl, \Yil).  
\end{equation}
Note that both distributions \eqref{eq:approx1} and \eqref{eq:approx2} are 
probability distributions, contrarily to what happens when relying on pseudo-likelihoods. 
Third, the choice of the fixed configuration $\btx$ relies on a
sequential Gibbs sampling from the approximate distribution~\eqref{eq:approx2}. With  these three tools at  hand, the algorithm consists in iteratively optimising with respect to $\psi=\{\psi_l\}_{1\le l \le L }$ the quantity 
\begin{equation*}
\esp^{2} \big[ \log  \pr^{1} (\bX |\psi, \btx) \big|
\psit, \btx,\bY \big],
\end{equation*}
computed with  the current value  of the parameter $\psit$  and current simulated  field  $\btx$. Here, $\esp^{2}$ denotes expectation under the probability distribution $\pr^{2}$.
This  quantity  should  be compared  to  the
original criterion~\eqref{eq:EM}.
\\

Let us now fully describe the procedure. 
For any current parameter value $\psit$ and fixed state value $\btx$, we let 
\[
\tilde Q(\psi| \psit,\btx) = \sum_{l=1}^L \sum_{i \in V^\star} \sum_{x\in \{0,1\}}\pr_{\psit} (\Xil =x|
\XNvl =\xtNvl, \Yil)\log \prpsi (\Xil =x|
\XNvl =\xtNvl) .
\]

The algorithm  consists in iterating  the following two steps  at time
$t$, 
\begin{itemize}
	\item \texttt{SE}-step: sequentially sample a configuration $\btx^{(t)}$ as follows
	for $1\le l \le L$ and $1\le i \le n$, sample $(\Xil)^{(t)}$ according to the conditional distribution 
	\[
	x \mapsto \pr_{\psi^{(t-1)}} \big(\Xil =x \big| \{\Xjl = (\xtjl)^{(t)} , j \in \Ni, j<i\}, \{\Xjl =(\xtjl)^{(t-1)} , j\in \Ni, j>i\}, \Yil\big) .
	\] 
	Thus, if $C_{il}=0$ we set $\Xil=0$ and whenever $C_{il}=1$,  we  sample the value $0$ with probability 
	\begin{equation}
	c \exp\Big( \betaa^{(t-1)}\sum_{j \in \Ni }
	A_{ij}^\star  C_{jl}\big[1\{(\xtjl)^{(t)}=0  , j<i\}  +1\{(\xtjl)^{(t-1)}=0,
	j>i\} \big] \Big) 1\{\Yil=0\} \label{eq:sample_0}
	\end{equation}
	and we sample the value $1$ with probability 
	\begin{align}
	c \exp\Big( \alpha_{i,l}^{(t-1)} + \betap^{(t-1)}\sum_{j \in \Ni }
	A_{ij}^\star  \big[1\{(\xtjl)^{(t)}=1  , j<i\}  +1\{(\xtjl)^{(t-1)}=1,
	j>i\} \big] \nonumber \\
	+ \Yil \log(p_{i,l}^{(t-1)}) + (1-\Yil)\log(1-p_{i,l}^{(t-1)}) \Big), \label{eq:sample_1}
	\end{align}
	where  $c$   is  a   normalising  constant  (set   such  that   the  2	probabilities sum to 1).

	\item \texttt{M}-step: 
 Optimize $ \tilde Q(\psi| \psit,\btx^{(t)})  $ with respect to
	$\psi= \{\alpha_{i,l}, \betaa,\betap\}_{i,l}$.  
\end{itemize}

We  now express  the quantity  $\tilde Q$  in our  model and  derive
update formulas in our model. 
First we set 
\begin{align*}
\tilde{p}_{i,l,t} (0)& = c \exp\Big( \betaa^{(t)}\sum_{j \in \Ni }
A_{ij}^\star C_{jl}1\{(\xtjl)^{(t)}=0 \} \Big) 1\{\Yil=0\} \\
\tilde{p}_{i,l,t} (1)& = c\exp\Big( \alpha_{i,l}^{(t)}+\betap^{(t)}\sum_{j \in \Ni }
A_{ij}^\star 1\{(\xtjl)^{(t)}=1 \} + \Yil \log(p_{i,l}^{(t)}) + (1-\Yil)\log(1-p_{i,l}^{(t)}) \Big) ,
\end{align*}
with  the  normalising   constant  $c$  such  that   $   \tilde{p}_{i,l,t} (0)  +
\tilde{p}_{i,l,t} (1)=1$.  Then the vector $( \tilde{p}_{i,l,t} (0) , \tilde{p}_{i,l,t} (1))$ is nothing
else than the probability distribution $\pr_{\psilt} (\Xil =\cdot| \XNvl
=\xtNvl, \Yil) $. 
From this quantity, we obtain 
\begin{align}
 \tilde Q(\psi|      \psit,\btx)  &   \nonumber \\
=       \sum_{l=1}^L 
 \sum_{i  \in      V^\star } C_{il}
\Big\{ &\tilde{p}_{i,l,t} (0)
\Big[\betaa\sum_{j \in     \Ni} A_{ij}^\star C_{jl} (1-\xtjl) \Big]
+ \tilde{p}_{i,l,t} (1)  \Big[\alpha_{i,l} +\betap\sum_{j \in     \Ni} A_{ij}^\star\xtjl \Big] \nonumber \\
-  &\log \Big[ \exp\Big(\betaa\sum_{j \in	\Ni}    A_{ij}^\star  C_{jl}  (1-\xtjl)  \Big)+ \exp \Big( \alpha_{i,l} + \betap \sum_{j \in	\Ni}    A_{ij}^\star    \xtjl   \Big) \Big]\Big\} .\label{eq:Q_tilde_2}
\end{align}
Optimising this quantity with respect to  
$\psi$ is done  numerically. To this aim, we provide below the derivatives of $\tilde Q$ wrt $\psi$.

Let us introduce the following quantities 
\begin{align*}
w^\star_i = & \sum_{j \in \Ni} A_{ij}^\star C_{jl},\\
w^\star_{i,l} = & \sum_{j \in \Ni} A_{ij}^\star  \xtjl
\end{align*}
which are the sum of weights of the neighbours of $i$ in $G^\star$ compatible with the location  $l$ and the sum of weights of the neighbours of $i$ in
$G^\star$ that are present at location $l$, respectively. Remembering that $C_{jl}\xtjl=\xtjl$, we have that
\[
\sum_{j \in \Ni} A_{ij}^\star C_{jl}( 1-\xtjl) =   w^\star_i  -  w^\star_{i,l} 
\]
is the sum of weights of the neighbours of $i$ in
$G^\star$ that are absent at location $l$ while compatible with that location.
We also use 
\begin{align*}
\text{den}_{i,l}(\betaa,\betap, \alpha_{i,l}) &= \exp[\betaa ( w^\star_i  -  w^\star_{i,l})]+\exp( \alpha_{i,l} + \betap   w^\star_{i,l} ).
\end{align*}
With these quantities at hand and relying on~\eqref{eq:Q_tilde_2}, we obtain 
\begin{align*}
\tilde Q(\psi| \psit,\btx)     &=        \sum_{l=1}^L \sum_{i \in      V^\star} C_{il}\Big\{ \tilde{p}_{i,l,t} (0) \betaa ( w^\star_i  -  w^\star_{i,l})
+ \tilde{p}_{i,l,t} (1) [ \alpha_{i,l} + \betap w^\star_{i,l}] \\
&-\log \text{den}_{i,l}(\betaa,\betap,  \alpha_{i,l})\Big\}.
\end{align*}

Let us recall that $\alpha_{i,l}$ is a shorthand for the quantity $a_i+a_l+ W_l^\intercal b_i + (W_l^2)^\intercal c_i$, so that we finally get, for each $1 \le l \le L$ and each $1\leq i \le n$, the derivatives  
\begin{align}
\frac{\partial \tilde Q} {\partial a_{i}} &= \sum_{l=1}^L C_{il}\Big[ \tilde{p}_{i,l,t} (1) - \frac{  \exp(\alpha_{i,l} + \betap w^\star_{i,l})} {\text{den}_{i,l}(\betaa,\betap,  \alpha_{i,l})}\Big] \label{eq:der_ai}\\
\frac{\partial \tilde Q} {\partial a_{l}} &= \sum_{i\in V^\star}C_{il} \Big[\tilde{p}_{i,l,t} (1) - \frac{  \exp(\alpha_{i,l} + \betap w^\star_{i,l})} {\text{den}_{i,l}(\betaa,\betap,  \alpha_{i,l})}\Big] \\
\frac{\partial \tilde Q} {\partial b_{i}} &= \sum_{l=1}^L C_{il} W_l ^\intercal \Big[\tilde{p}_{i,l,t} (1) - \frac{  \exp(\alpha_{i,l} + \betap w^\star_{i,l})} {\text{den}_{i,l}(\betaa,\betap,  \alpha_{i,l})} \Big]\\
\frac{\partial \tilde Q} {\partial c_{i}} &= \sum_{l=1}^L C_{il} (W_l^2)^\intercal \Big[\tilde{p}_{i,l,t} (1) - \frac{  \exp(\alpha_{i,l} + \betap w^\star_{i,l})} {\text{den}_{i,l}(\betaa,\betap,  \alpha_{i,l})} \Big]\\
    \frac{\partial \tilde Q} {\partial \betaa} &=  \sum_{i \in      V^\star}C_{il}\left[\tilde{p}_{i,l,t} (0)( w^\star_i  -  w^\star_{i,l}) - \frac{ ( w^\star_i  -  w^\star_{i,l})\exp[\betaa ( w^\star_i  -  w^\star_{i,l})]} {\text{den}_{i,l}(\betaa,\betap,  \alpha_{i,l})}\right]
    \\
      \frac{\partial \tilde Q} {\partial \betap} &= 
       \sum_{i \in      V^\star}C_{il}\left[\tilde{p}_{i,l,t} (1) w^\star_{i,l} - \frac{   w^\star_{i,l} \exp[\alpha_{i,l} + \betap w^\star_{i,l}]} {\text{den}_{i,l}(\betaa,\betap,  \alpha_{i,l})} \right].
       \label{eq:der_betap}
\end{align}

The simulated field algorithm is described in Algorithm~\ref{algo:SEM}. 

\begin{algorithm}[ht]
  \caption{Simulated field algorithm}
  \label{algo:SEM}
\begin{algorithmic}
 \State {\bfseries Input:} Observed presence/absence data $\bY$, adjacency matrix of metanetwork $A^\star$.
 \State {\bfseries Initialization:} Choose   initial  values $\btx^{(0)}, \psi^{(0)}$. 
\State Set $t=1$.
  \While{not converged}
  \State {\bfseries Simulation step:}
  \For{$1\le l\le L $}
    \For{$1\le i \le n$}
  \State  Sample $(\xtil)^{(t)}$ from $\{0,1\}$ 
  relying on the vector of probabilities~\eqref{eq:sample_0} and~\eqref{eq:sample_1}.  
  \EndFor
  \EndFor
  \State Compute $\tilde Q(\psi | \psit ; \btx)$ from~\eqref{eq:Q_tilde_2}.
  \State {\bfseries Maximization step:}
  \State Compute the value $\hat \psi$ zeroing the derivatives \eqref{eq:der_ai}--\eqref{eq:der_betap}. 
 \State Update  parameter  $\psi^{(t)}=\hat \psi$.  
 \State Increment $t$.
  \EndWhile
\end{algorithmic}
\end{algorithm}

\begin{rem} 
In the case with no sampling effects (namely $p_{i,l}=1$), the simulation step is skipped (since $\bX=\bY$), the
quantities $\tilde{p}_{i,l,t}$ become  $\tilde{p}_{i,l,t}(x)=1\{\Xil =x\}$  and the criteria to optimize reduces to the quantity
\[
  \tilde Q_{direct}(\psi ) = \sum_{l=1}^L \sum_{i \in V^\star} \sum_{x\in \{0,1\}}
  \log \prpsi (\Xil =x|
\XNvl) .
\]
This means that in this specific case, our method consists exactly in a pseudo-likelihood estimation, which is known to be consistent as the number of observations increases~\citep{Besag75}. 
Therefore, the estimation algorithm is more computationally affordable in this case since it consists in a simple iteration of the \texttt{M}-step (i.e. the 'maximization step' in Algorithm~\ref{algo:SEM}).
\end{rem}

\subsection{Additional details on the implementation}
The 'maximization step' in Algorithm~\ref{algo:SEM} is performed using the vector Broyden-Fletcher-Goldfarb-Shanno (BFGS) algorithm implemented in the GNU Scientific Library (
\url{https://www.gnu.org/software/gsl/}).
We observed that this algorithm was sensitive to the initial value of the parameters. After analyzing synthetic datasets simulated from the model and estimating the model with various initial values, we validated the following combination of initial parameters:
\begin{align*}
    a_i=a_l=\frac{a_0}{2}\\
    b_i=c_i=0\\
    \betaa=\betap=0
\end{align*}
with $a_0 = \log(\frac{\bar Y}{1- \bar Y})$ and $\bar Y = \sum_{il} Y_{il} /(nL)$.

\section{ \blue{Simulation under ELGRIN model and inference}}
\label{SI_sec:samp_resamp}
\blue{In order to test the statistical performance of ELGRIN model, we simulated under ELGRIN model and tried to recover the sample parameters. Once the parameters inferred, we simulated new data using ELGRIN again while relying on these inferred values. We then inferred the parameters of this new dataset to test the stability of parameters inference under resampling. HTLM vignette is available at  \url{https://plmlab.math.cnrs.fr/econetproject/econetwork/-/blob/master/vignettes/simul_under_elgrin.html}. \\
We chose the same metaweb and environmental gradient as in the colonisation-extinction simulation (see section \ref{SI_sec:CE}), with $50$ species and used $400$ sites. We draw $a_i$ and $a_l$ uniformly in $[-0.25,0.25]$, $b_i$ and $b_l$ uniformly in $[-0.5,0.5]$. We chose a gradient of $\beta_{l,co-abs}$ and $\beta_{l,co-pres}$ ranging from $0$ to $1$. \\
We represent the comparison of the original model and the inferred model in Figs. \ref{fig:comp_mod1} and \ref{fig:comp_mod2}. The parameters sampled with ELGRIN are reasonably recovered by the inference algorithm and stable under another sampling and inference step.
}

\begin{figure}
    \centering
    \includegraphics[width = \textwidth]{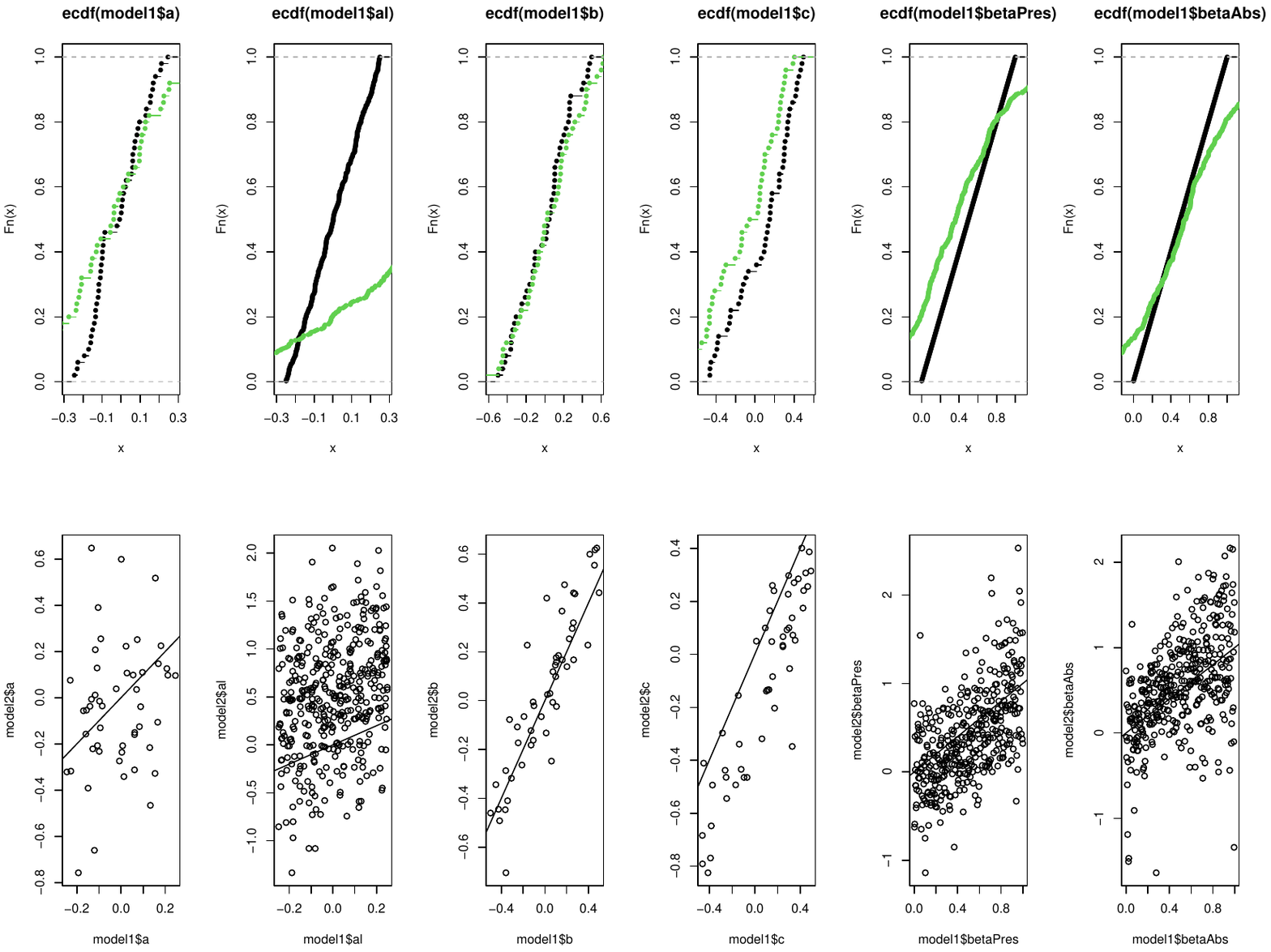}
    \caption{\blue{Comparison between the sample parameters and the inferred parameters using ELGRIN inference algorithm. Top: Cumulative distribution of the sample parameters (model $1$, black) and the inferred ones (model $2$, green). Bottom: Scatter plot of pairs of parameters (sample and inferred) and the diagonal axis.}}
    \label{fig:comp_mod1}
\end{figure}

\begin{figure}
    \centering
    \includegraphics[width = \textwidth]{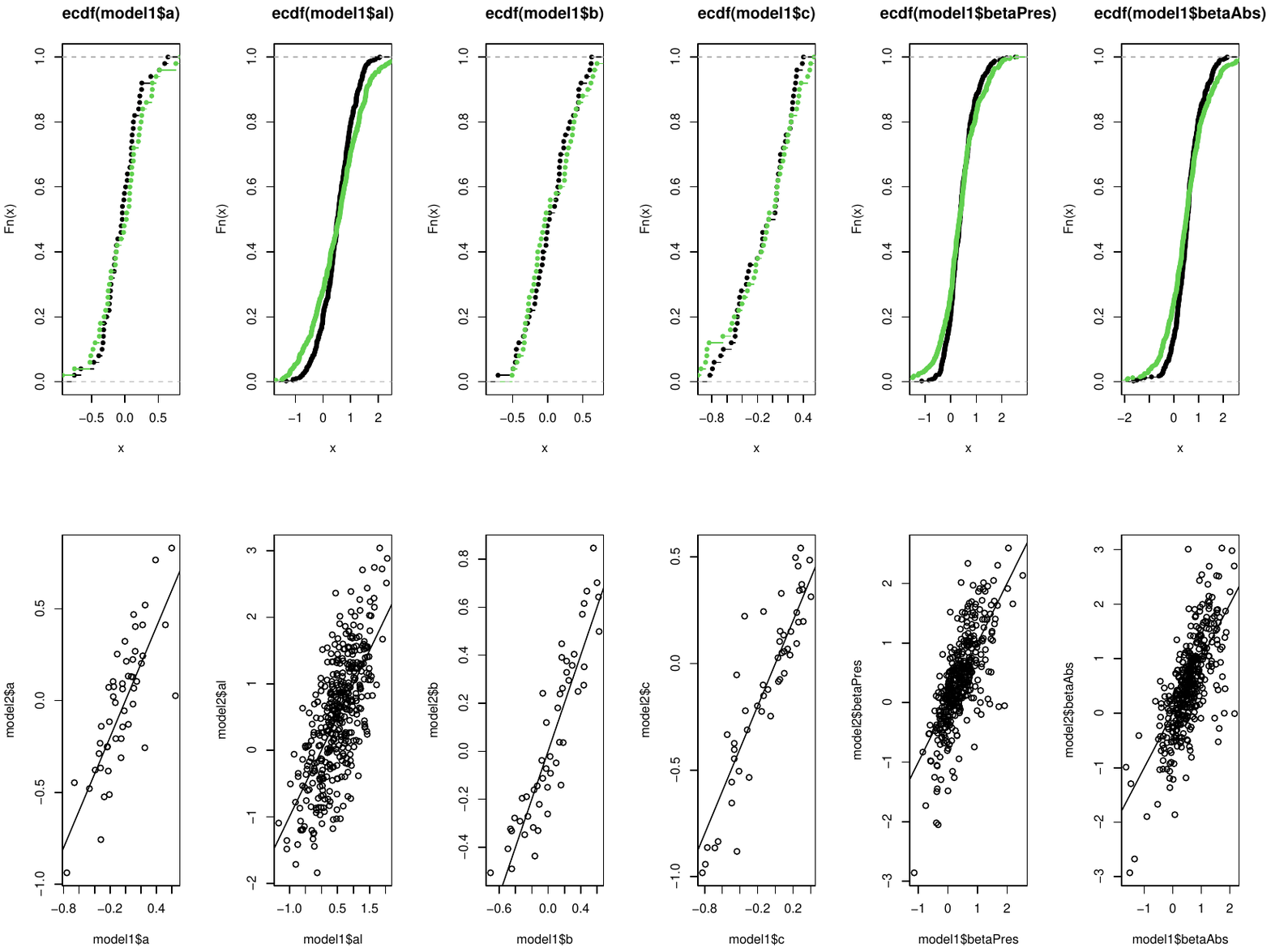}
    \caption{\blue{Comparison between the re-sample parameters (used to sample under ELGRIN model) and the inferred parameters using ELGRIN inference algorithm. Top: Cumulative distribution of the re-sample parameters (model $1$, black) and the inferred ones (model $2$, green). Bottom: Scatter plot of pairs of parameters (re-sample and inferred) and the diagonal axis.}}
    \label{fig:comp_mod2}
\end{figure}

\section{Simulations with three different theoretical models}
We provide here models and details on the three simulations of the paper. HTLM vignettes for the three simulations are available at \url{https://plmlab.math.cnrs.fr/econetproject/econetwork}. For each model, we simulated three different scenarios: positive (i.e. mutualism), negative (i.e. competition) or no interactions. The scenario with no interactions uses an empty metanetwork  to generate the data. However, inference with ELGRIN in this case relied on a metanetwork with interactions (to be specified below). 

In the following, when ELGRIN is fitted on a dataset, the inference procedure outputs estimated parameters values. For any species $i$, its niche optima was estimated from these values, relying on the optimum of the estimated function $w \mapsto \hat a_i + \hat b_i w +\hat c_i w^2$ within the interval $[\omega_i,\Omega_i]$
 defined in \eqref{SI_eq:interval} (here dimension $d=1$).

\subsection{Lotka-Volterra model: details and simulation set-up}
\label{SI_sec:LV}
We sampled species communities from the equilibrium of a deterministic Lotka-Volterra model \citep{Takeuchi}. 
We defined the environmental niche of each species as a Gaussian distribution centered on a given optimum. The environmental niches optima were evenly taken on a grid whereas the standard deviations were all equal to a given value  $\sigma$ for simplicity.

\subsubsection*{Building the network from niche values}
We constructed the metanetwork $G^\star$ used for generating the data in scenarios with interactions (positive and negative) and later used for inference with ELGRIN in the three scenarios (i.e. including when there are no interactions). Let $\mu_i$ and $\mu_j$ be the niche optima of two distinct species. We sampled symmetric interaction between species $i$ and $j$ according to a Bernoulli law of parameter $\lambda m| \mu_i - \mu_j |^{-1}$, where $m=\max_{i,j} \left (|\mu_i - \mu_j |^{-1} \right )$ and $\lambda$ is a parameter modulating the overall edge number. We obtained a metanetwork $G^\star$  symmetric with no self-loops.

\subsubsection*{Modelling the dynamics}
We assume, for species $i$, a per-capita growth rate $r_i(w)$ depending on the environment value $w$ and following a Gaussian function of mean $\mu_i$. We then model $N_{iw}(t)$, the abundance of species $i$ at environment value $w$ and time $t$, using a generalised Lotka-Volterra dynamical model with intraspecific competition. In the negative interactions scenario, we used  
\begin{equation}
\frac{1}{N_{iw}}\frac{dN_{iw}}{dt} = r_i(w) - \sum_{j} C_{ij} N_{jw},
\label{lvEq}
\end{equation}
where $C_{ij} = A^\star_{ij} + c\1(i=j)$
with $A^\star$ the adjacency matrix of $G^{\star}$ and $c$ the intraspecific competition coefficient.
For the positive interaction scenario, we used
\begin{equation}
    \label{eqref:LV_pos}
\frac{1}{N_{iw}}\frac{dN_{iw}}{dt} = r_i(w) + \sum_{j} M_{ij} N_{jw},
\end{equation}
where $M_{ij} = A^\star_{ij}/M_0 - c\1(i=j)$
where $M_0$ is a constant ($M_0>1$) that reduces the strength of positive interactions in order to get convergence towards a finite abundance value.
In the no interactions scenario, we use  $C_{ij}=M_{ij}=0$ for all $i,j$ in the above equations. Note though that we used the simulated metanetwork $G^\star$ for inference with ELGRIN in the three scenarios. 
From the equilibrium point $\mathbf{N}^{\star}_w=(N^\star_{iw})_i$ (with $N^\star_{iw}=\lim_{t\to +\infty} N_{iw}(t)$, limit that is assumed here to be unique and independent of initial conditions), we sample presence or absence $\Xil$ of each species $i$ at location $l$ using a Bernoulli law of parameter $\min(1,N^\star_{i w_l}/5)$.

\subsubsection*{Parameter values}
We performed simulations with $N=50$ species and $L=400$ locations. The environmental niches optima were evenly taken on a grid between $-2$ and $2$ whereas the environmental gradient ranged from $-3$ to $3$. We set the standard deviations of niche distributions to $\sigma=1$ and we set the intraspecific competition term to $c = 1/10$ for all species.  The constant $M_0$ is set to 50. We ran the simulation of the Lotka-Volterra  dynamics for $10,000$ time steps. Fig.~\ref{fig:niches_net_LV} shows growth rates in function of the environment and metanetwork. We also represented the distribution of species presence-absences and species richness 
under the three interaction scenarios in Fig.~\ref{fig:spp_distrib_LV} and Fig.~\ref{fig:simu_LV_richness}. Niche optima inferred from ELGRIN on this dataset are shown in Fig.~\ref{fig:simu_LV_fit_niches}. Association parameters $\betaa$ and $\betap$ are represented in the main text, Fig.~\ref{fig:simu_LV}.

\begin{figure}[ht!]
	\begin{center}
		\includegraphics[height=10cm]{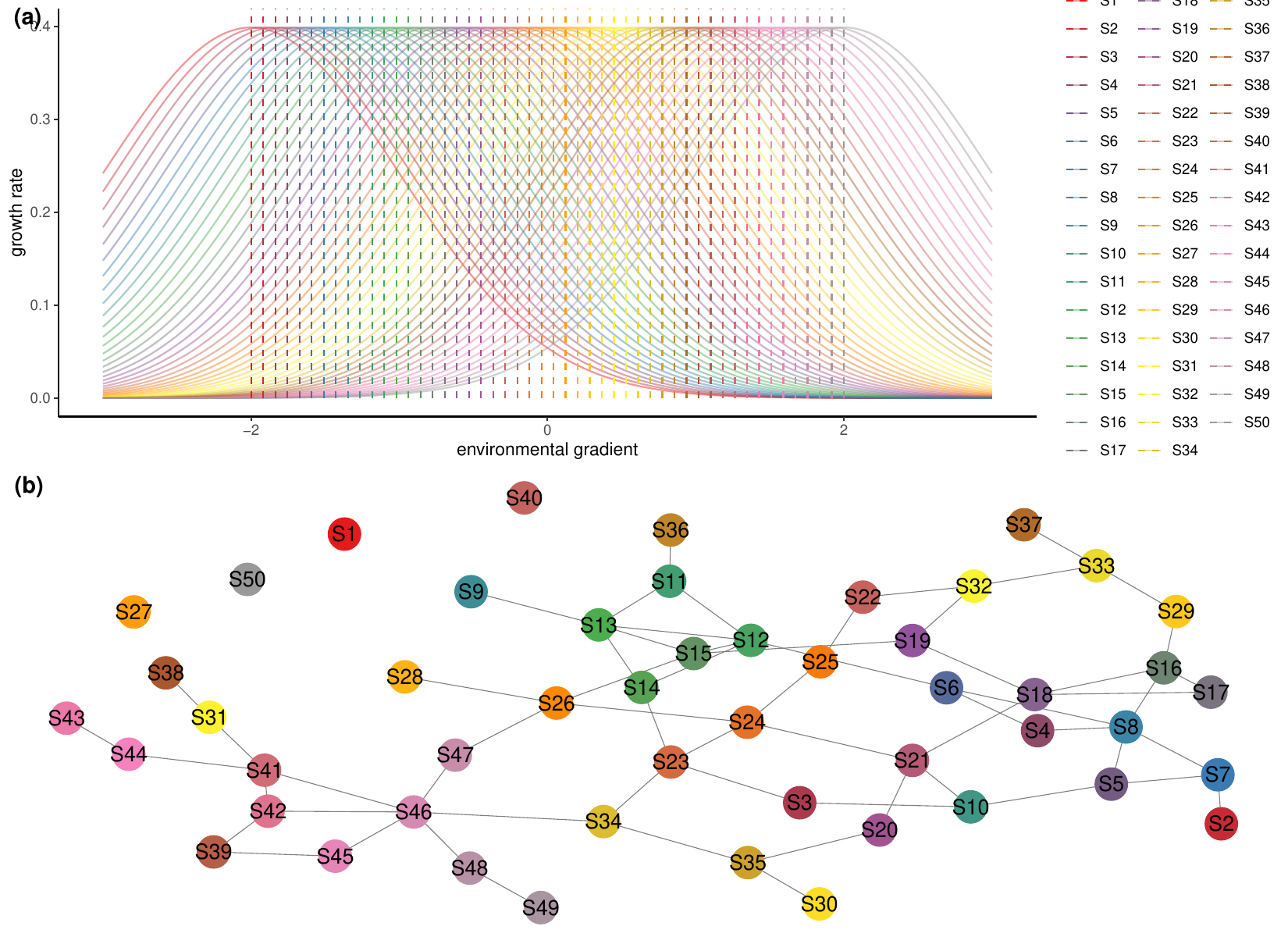}
		\caption{Simulations under Lotka-Volterra and colonisation-extinction models. (a) Growth rates in function of the environment for the $50$ considered species. (b) Representation of the metanetwork used for simulations in the two scenarios with interactions and for estimation with ELGRIN in the three scenarios. Nodes are colored according to the value of niche optima along the environmental gradient.}
		\label{fig:niches_net_LV}
	\end{center}
\end{figure}

\begin{figure}[ht!]
	\begin{center}
		\includegraphics
		[width=0.6\textwidth]{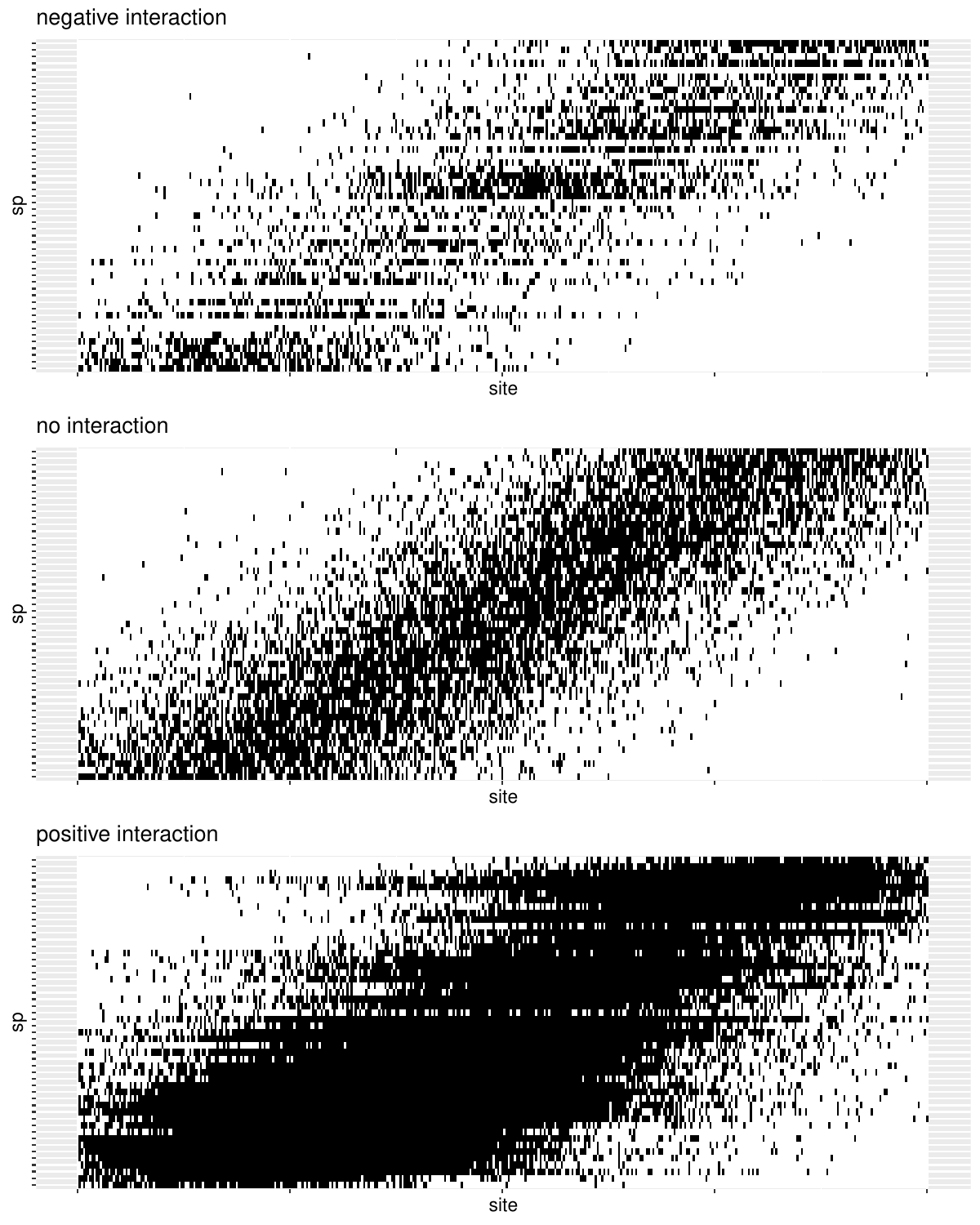}
		\caption{Presence-absence of species ($y$-axis) along the environmental gradient ($x$-axis) for the Lotka-volterra simulations across the three interaction scenarios.}
		\label{fig:spp_distrib_LV}
	\end{center}
\end{figure}

\begin{figure}[ht!]
	\begin{center}
		\includegraphics
		[width=0.6\textwidth]{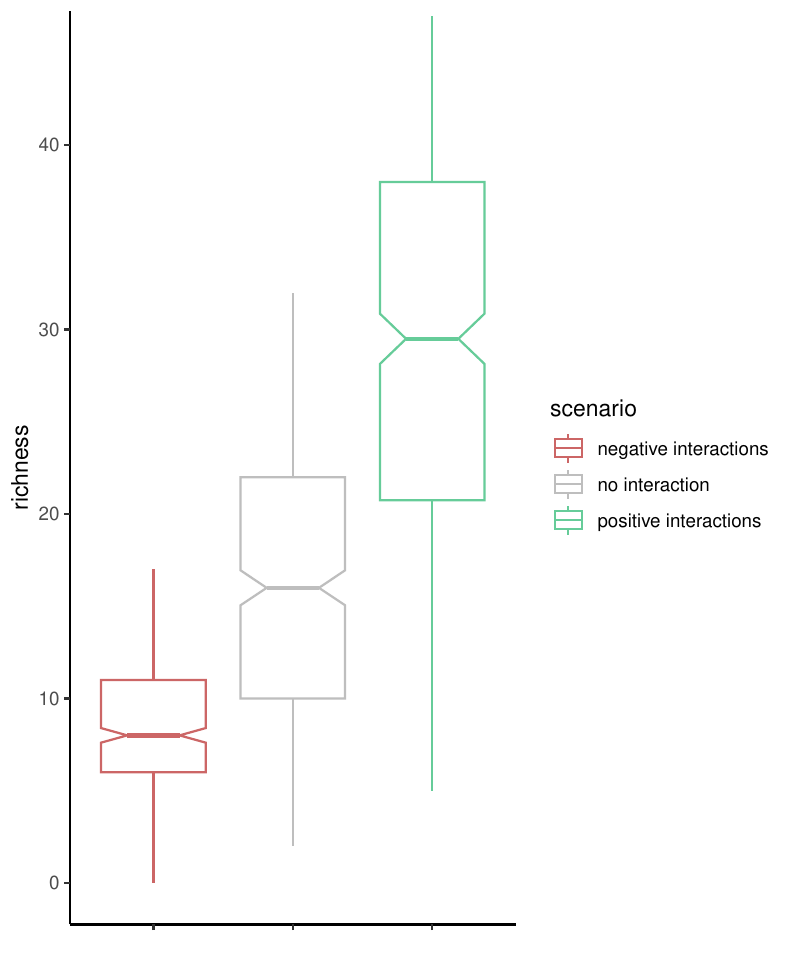}
		\caption{Distribution of species richness (observed number of present species) for the Lotka-Volterra simulation under the three interaction scenarios.}
		\label{fig:simu_LV_richness}
	\end{center}
\end{figure}

\begin{figure}[ht!]
	\begin{center}
		\includegraphics
		[width=0.6\textwidth]{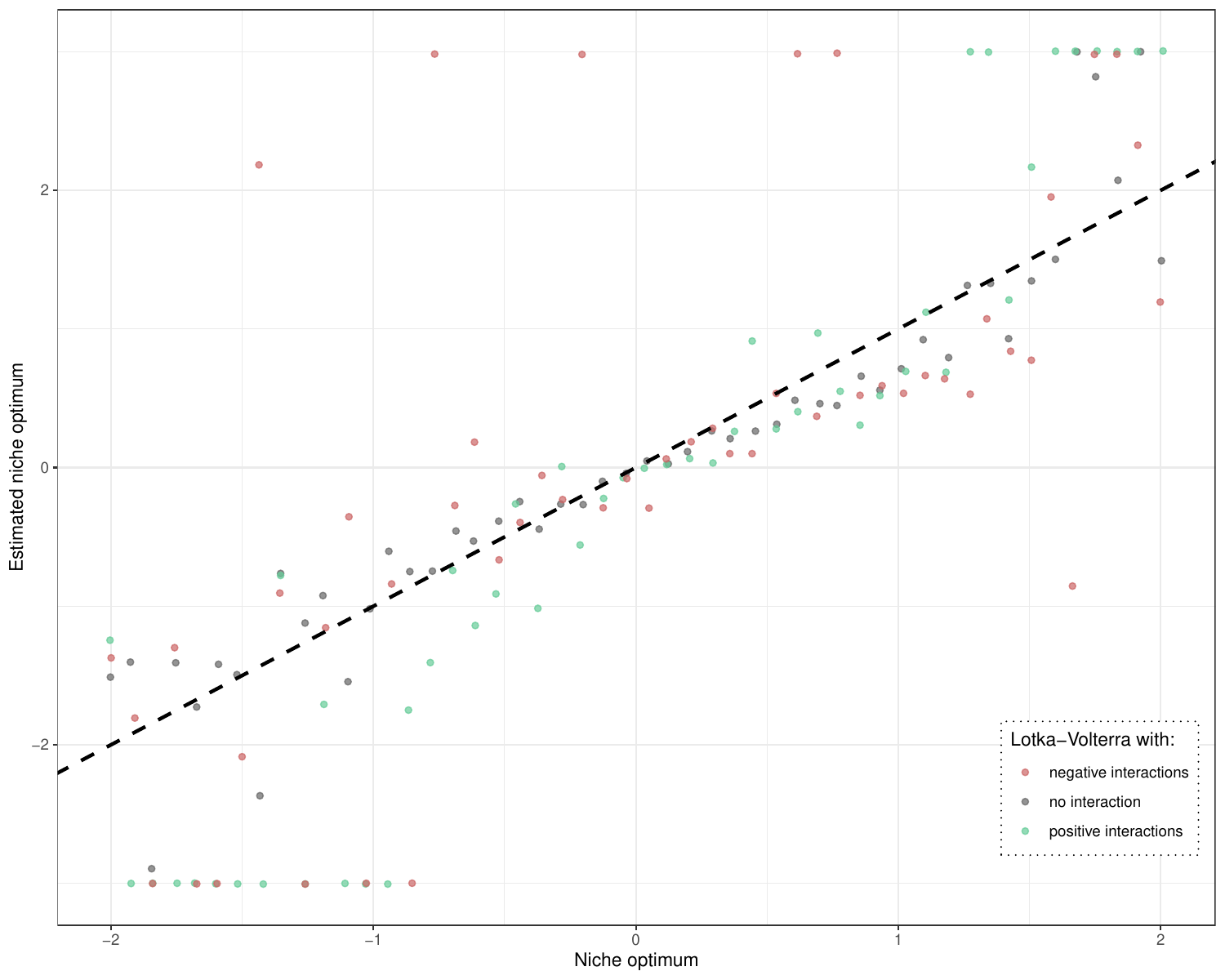}
		\caption{Estimated niche optima versus true niche optima for the Lotka-Volterra simulation under the three interaction scenarios.}
		\label{fig:simu_LV_fit_niches}
	\end{center}
\end{figure}

\subsubsection*{Results and discussion}

We notice that species richness increases from the competition to the mutualistic scenarios, as positive interactions enhance the possibility of species to be present (vice-versa for competition). We see that except for the positive interactions scenario, ELGRIN reasonably infers niche optima and association parameters ($\betaa$ and $\betap$, as shown in Fig. 2 in the main text) on this community data built from Lotka-Volterra model (see the discussion in the main text for further insights). We however acknowledge a large variance on association parameters for the negative interaction model and as already underlined in the main text, the inability of ELGRIN to identify the positive interactions scenario.  
We remark that this positive interaction scenario of Lotka-Volterra model is a particularly harsh test for ELGRIN. Indeed, positive interactions increase the effective growth rate, leading to the risk of explosion of the system. For this reason, we were obliged to reduce the overall interaction strength when simulating the data (parameter $M_0$ in Equation~\eqref{eqref:LV_pos}), thus reducing their signal in resulting data. Moreover, positive interactions cause species to be present everywhere along their fundamental niche (e.g.	Fig.~\ref{fig:spp_distrib_LV}), so that their distribution can be completely explained by the Grinellian part of the model. In other words, the data look exactly as if they were obtained from a Grinellian model, where only the environmental variables shape the species distribution. 
Therefore, no signal is left for the Eltonian part, and the association parameters are inferred to be close to zero.

\subsection{Colonisation-extinction model: details and simulation set-up}
\label{SI_sec:CE}
We sampled species communities from the stationary distribution of a stochastic colonisation-extinction model (see \citealt{ohlmann2022assessing}). 
We  kept the same environmental gradient, niches and metanetwork  as in the Lotka-Volterra simulation. We also  combined this model with  the three interaction scenarios: negative interactions (i.e. competition), positive interactions (i.e. mutualism) or no interactions.

\subsubsection*{Modelling the dynamics}
We note $X_{i}^t$ the binary random variable associated to presence of species $i$ at discrete time $t$ and $w$ the value of the environmental gradient. We model niche and interaction effects trough conditional probabilities.

\paragraph{Colonisation-extinction model without interaction}
In this scenario, we assume that interaction effects do not impact colonisation or extinction. Extinction probability $p_e$ is constant whereas colonisation probability $c_i(w)$ for species $i$ depends on the value of the environmental gradient $w$ only 
\begin{equation*}
      P(X_{i}^{t+1} = 1|\{X_j^t\}_j,w) = P(X_{i}^{t+1} = 1|X_{i}^t,w) \propto c_i(w)(1-X_{i}^t) + (1-p_e)X_{i}^t,
\end{equation*} 
where $c_i(w)$ is a Gaussian function with mean $\mu_i$ and variance $\eta^2$ and $\propto$ means up to a normalizing constant.
We simulated this Markov chain and sampled from the stationary distribution to generate a joint species distribution.

\paragraph{Colonisation-extinction model with positive and negative interactions}
In these scenarios, we assume that both abiotic niche effects and interspecific interactions do impact colonisation-extinction processes. Environmental gradient modulates colonisation probability whereas interactions modulate both colonisation and extinction probabilities.

For the positive interactions scenario, we have:
\begin{align*}
     P(X_{i}^{t+1} = 1|\{X_j^t\}_j,w) &= P(X_{i}^{t+1} = 1|X_{i}^t,X^t_{N(i)},w) \\
     & \propto c_i(w)\exp\Big(\frac{\sum_{k \in N(i)} X_k^t}{|N(i)|}\Big)(1-X_{i}^t) + \Big[1-p_e\exp\Big(-\frac{\sum_{k \in N(i)} X_k^t}{|N(i)|}\Big)\Big]X_{i}^t,
\end{align*}
and for the negative interactions scenario
\begin{align*}
P(X_{i}^{t+1} = 1|\{X_j^t\}_j,w) &=
      P(X_{i}^{t+1} = 1|X_{i}^t,X^t_{N(i)},w) \\
      &\propto c_i(w)\exp\Big(-\frac{\sum_{k \in N(i)} X_k^t}{|N(i)|}\Big)(1-X_{i}^t) + \Big[1-p_e\exp\Big(\frac{\sum_{k \in N(i)} X_k^t}{|N(i)|}\Big)\Big]X_{i}^t, 
\end{align*}
where $N(i)$ is the set of neighbour species of species $i$ in the metanetwork. Similarly as for the no  interaction scenario, we sampled the species co-occurrences in the stationary distributions of each of these scenarios.

\subsubsection*{Parameter values}

We performed simulations with $N=50$ species and $L=400$ locations. Extinction probability was set to $p_e=2\%$ and colonisation probability $c_i(w)$ is Gaussian with mean $\mu_i$ and standard deviation $\eta=1$. We ran each simulation dynamics for $3,000$ time steps. We represented the distribution of species presence-absences and species richness 
under the three interaction scenarios in Fig.~\ref{fig:spp_distrib_CE} and Fig.~\ref{fig:simu_CE_richness}.
Niche optima inferred from ELGRIN on this dataset are shown in Fig.~\ref{fig:simu_CE_fit_niches}. 
Association parameters $\betaa$ and $\betap$ are represented in the main text, Fig.~\ref{fig:simu_CE}.

\begin{figure}[ht!]
	\begin{center}
		\includegraphics
		[width=0.6\textwidth]{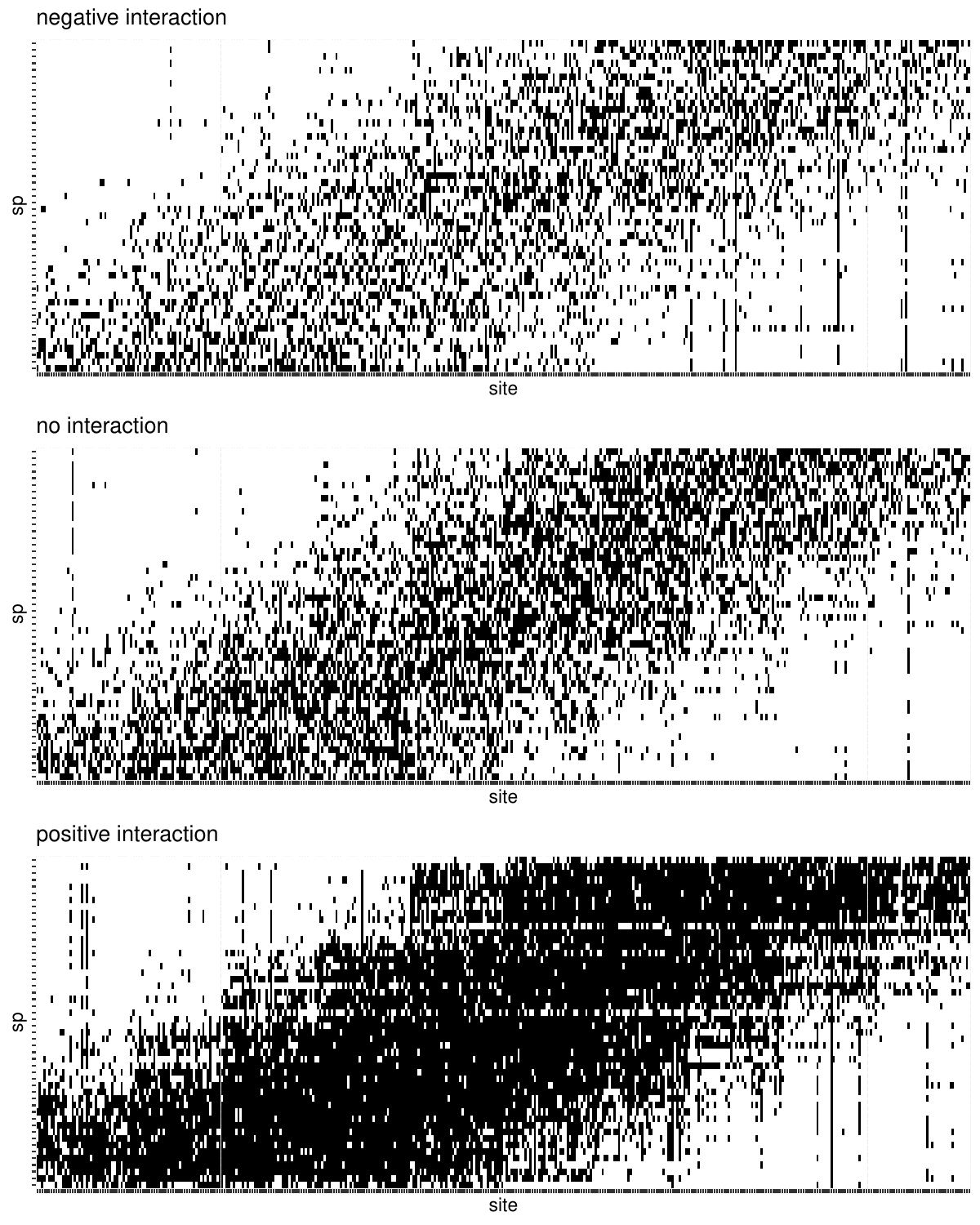}
		\caption{Presence-absence of species ($y$-axis) along the environmental gradient ($x$-axis) for the colonisation-extinction simulations, across the three interaction scenarios.}
		\label{fig:spp_distrib_CE}
	\end{center}
\end{figure}

\begin{figure}[ht!]
	\begin{center}
		\includegraphics
		[width=0.6\textwidth]{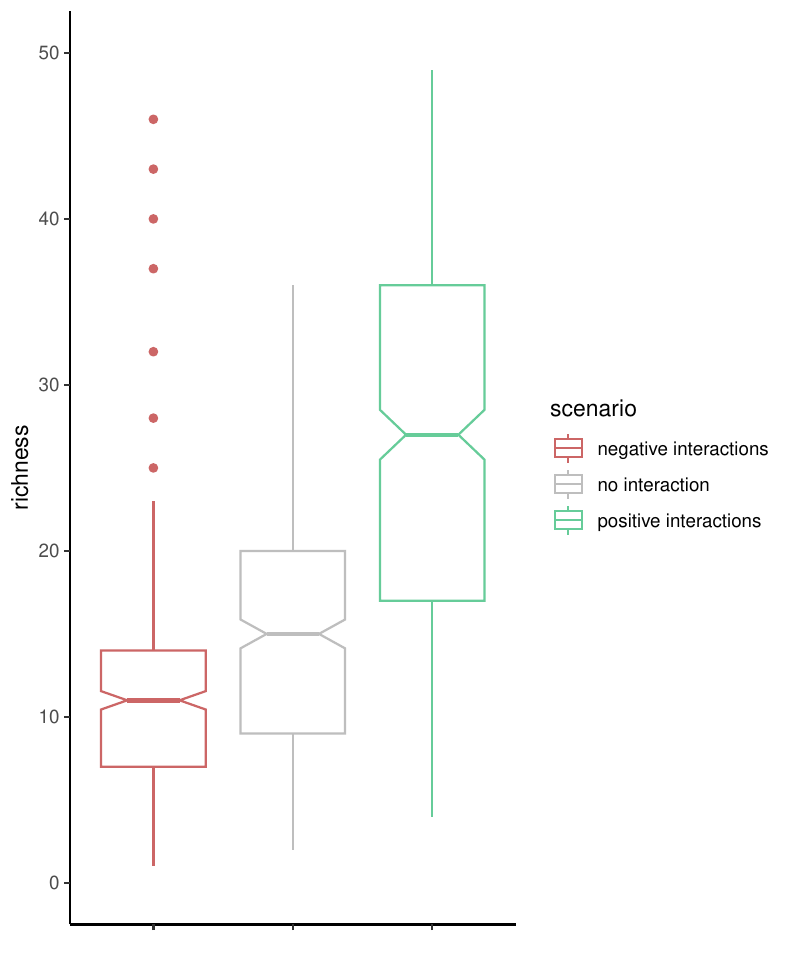}
		\caption{Distribution of species richness (observed number of present species) for the colonisation-extinction simulation under the three interaction scenarios.}
		\label{fig:simu_CE_richness}
	\end{center}
\end{figure}

\begin{figure}[ht!]
	\begin{center}
		\includegraphics
		[width=0.6\textwidth]{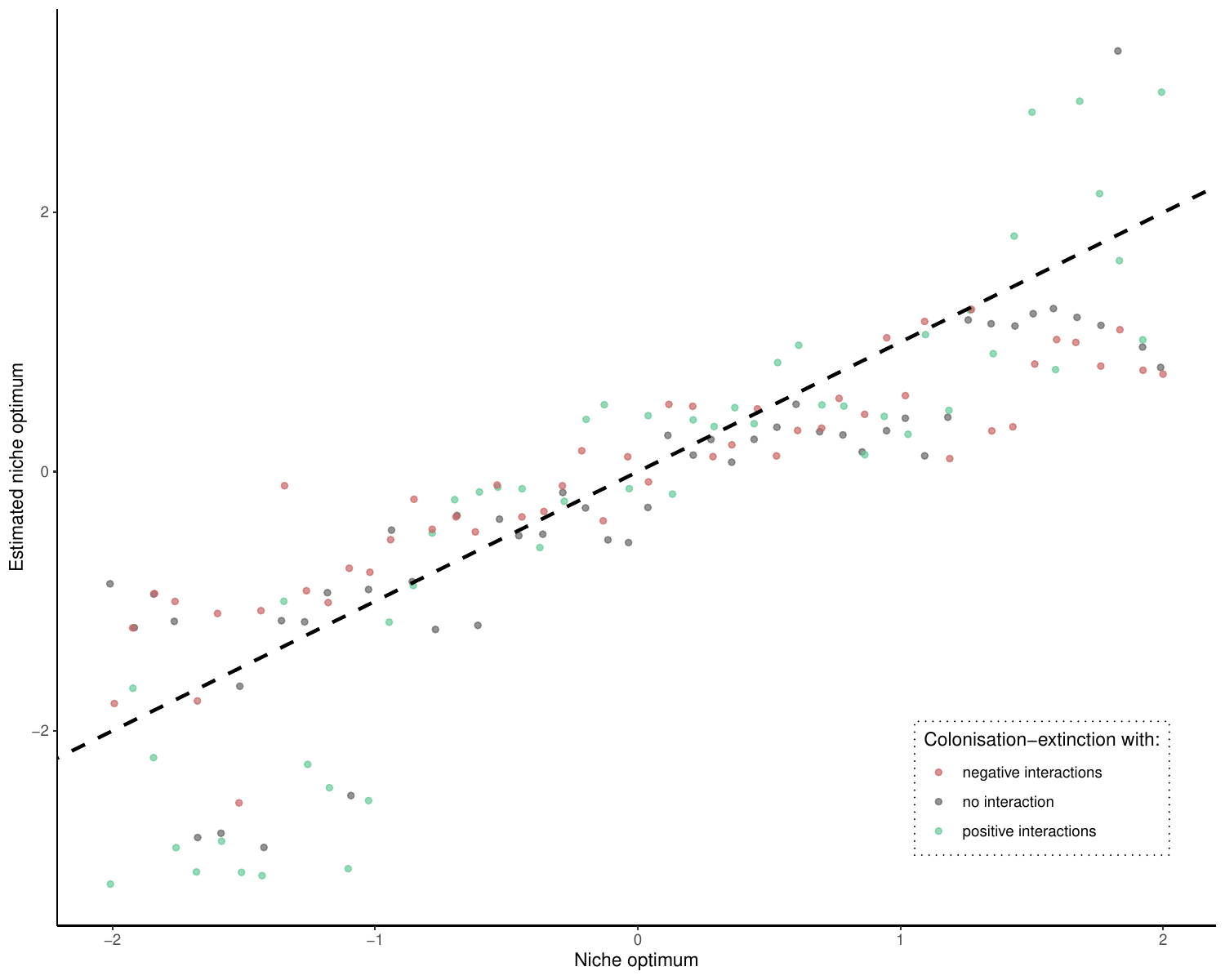}
		\caption{Estimated niche optima versus true niche optima for the colonisation-extinction simulation under the three interaction scenarios.}
		\label{fig:simu_CE_fit_niches}
	\end{center}
\end{figure}

\subsubsection*{Results and discussion}

We notice that species richness increases from the competition to the mutualistic scenarios, as positive interactions enhance the possibility of species to be present (vice-versa for competition, Fig.~\ref{fig:simu_CE_richness}).
For each scenario, the distribution of association parameters ($\betap$ and $\betaa$, see Fig.3 in the main text) have a negative median for negative interactions, a median close to zero for the case without interaction and a positive median for positive interactions. The sign of inferred (static) association parameters is the same as the sign of dynamic interaction parameters. 
 
Moreover, ELGRIN correctly infers niche optima in the three interaction scenarios (Fig~\ref{fig:simu_CE_fit_niches}). Consequently, on these simulations, ELGRIN separates environmental effects for biotic interactions (see the discussion in the main text for further insights).

\subsection{VirtualCom model: details and simulation set-up}
\label{SI_sec:VC}
We considered $N$ species in the species pool and $L$ communities to simulate (i.e. the number of locations).  
We defined the environmental niche (or preference) of each species as a Gaussian distribution centered on a given optimum. The environmental niches optima were regularly taken on a grid between -2 and 2, whereas the standard deviations were all equal to a given value  $\sigma$ for simplicity. 
Each community or location $l$ has the same carrying capacity $K$ (i.e. the exact number of individuals in each location). 

\subsubsection*{Building the interaction networks from niche values}
Here we constructed two different metanetworks $G^\star$ used to simulate data in the two interaction scenarios. Let $\mu_i$ and $\mu_j$ be the niche optima of two species and $\sigma$ the standard deviation of their niche. We considered that the two considered species potentially interact in the mutualistic metanetwork if  $\sigma <  | \mu_i - \mu_j| < 2 \sigma$. Regarding competition, we considered that the two species potentially compete if they share the same environmental niche, and thus if  $| \mu_i - \mu_j| < \sigma$.
Among all potential species interactions, we randomly sampled $50\%$ of them for both competition and mutualism. Inference with ELGRIN in the scenarios with interactions relied on the respective metanetworks used for simulation. In the no interaction scenario, ELGRIN inference relied on the positive interactions metanetwork (corresponding to mutualism).

\subsubsection*{Modelling the dynamics}
The community assembly process was randomly initialized with a set of individuals that were randomly selected in the species pool until the carrying capacity $K$ was reached. At each time step, the probability of an individual from species $i$ to replace a random individual of the community $l$ is $R_{il}$. This probability depends on how the environmental conditions at location $l$ are suitable for species $i$ (environmental filter) and on the number of individuals present in community $l$ that interact with species $i$ (competition or mutualism filter). More precisely, we consider the following equation defining the relative importance of environmental and biotic filters respectively:
$$
R_{il} \propto \exp \left[ \gamma_{env} \log(p_{il}^{env}) \ +  \gamma_{inter} \log(p_{il}^{inter}) \right],
$$
where $\gamma_{env}$ and $\gamma_{inter}$ are tuning parameters giving weights to abiotic and biotic components, and $p_{il}^{env}$ and $p_{il}^{inter}$ are probabilities of species replacement with different filters.
The probability $p_{il}^{env}$ accounts for the environmental filtering and is a rescaled density of the Gaussian niche of species $i$ at the environmental value of location $l$ (the scaling ensures this value ranges in $[0,1]$).
When the environment in community $l$ is suitable to species $i$, the probability that this species enters this community becomes high. 

We then have a term dealing with species interactions. In the no interaction scenario, the constant  $\gamma_{il}$ is set to 0.
Otherwise, the interaction term is set  as
$$
p_{il}^{inter} = \left\{ 
\begin{array}{ll}
	K^{-1}  
	\sum_{j ; (i,j)\in E^\star}
	K_{jl} & \text{ for mutualism,} \\
	1 - K^{-1}
	\sum_{j ; (i,j)\in E^\star}
	K_{jl} & \text{ for competition,}
\end{array}
\right. 
$$
where $K_{jl}$ is the number of individuals of species $j$ in community $l$, such that the total carrying capacity $K =  \sum_j K_{jl}$. 
In case of mutualism, the larger number of individuals of species connected with $i$ in the metanetwork are present in location $l$, the higher is the probability of an individual of species $i$ to enter the community.  For competition, the opposite effect is induced.
The tuning parameters $\gamma_{env}$ and $\gamma_{inter}$ weight the relative importance of the different filters. This algorithm updates the communities until an equilibrium is reached. To assess the equilibrium state, we calculated the Shannon diversity for each location over time, and checked for convergence. Lastly, we deduced species presence/absence  by examining species composition in each location.

\subsubsection*{Parameter values}

We performed simulations with $N=50$ species and $L=400$ locations, with a carrying capacity of $K=40$ individuals. The standard deviations of the Gaussian niche distributions were set to  $\sigma=1$ for all species. We chose $\gamma_{env}=1$ and $\gamma_{metanetwork}=10$ in case of competition and  $5$ in case of mutualism.  Fig.~\ref{fig:niches_net_VC} shows growth rates in function of the environment and the two metanetworks (for positive and negative interactions).  We simulated 100 time steps such that the algorithm convergence was achieved in practice. We repeated the whole procedure 10 times and verified that we obtained equivalent qualitative results. Simulations were implemented with R version 3.6.2 and a modified version of the VirtualCom package. We represented the distribution of species presence-absences and species richness 
under the three interaction scenarios in Fig.~\ref{fig:spp_distrib_VC} and Fig.~\ref{fig:simu_VC_richness}.
Niche optima inferred from ELGRIN on this dataset are shown in Fig.~\ref{fig:simu_VC_fit_niches}. 
Association parameters $\betaa$ and $\betap$ are represented in the main text, Fig.~\ref{fig:simu_VC}.

\begin{figure}[H]
	\begin{center}
		\includegraphics[width=0.6\textwidth]{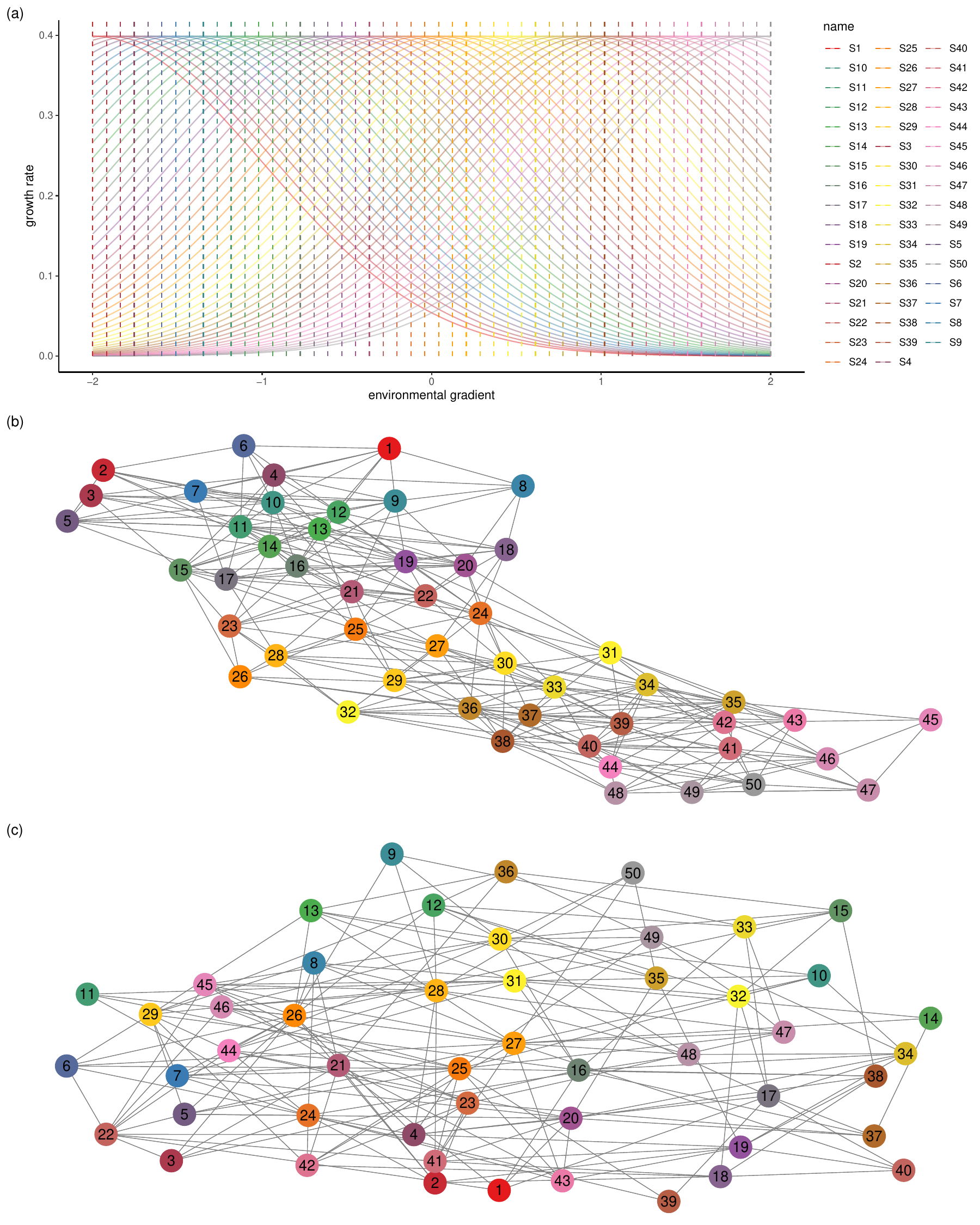}
		\caption{Simulations under VirtualCom model. (a) Growth rates in function of the environment for the $50$ considered species. Representation of the metanetworks used for simulations of VirtualCom in the facilitation case (b) and in the competition case (c). Nodes are colored according to the value of niche optima along the environmental gradient.}
		\label{fig:niches_net_VC}
	\end{center}
\end{figure}

\begin{figure}[H]
	\begin{center}
		\includegraphics
		[width=0.6\textwidth]{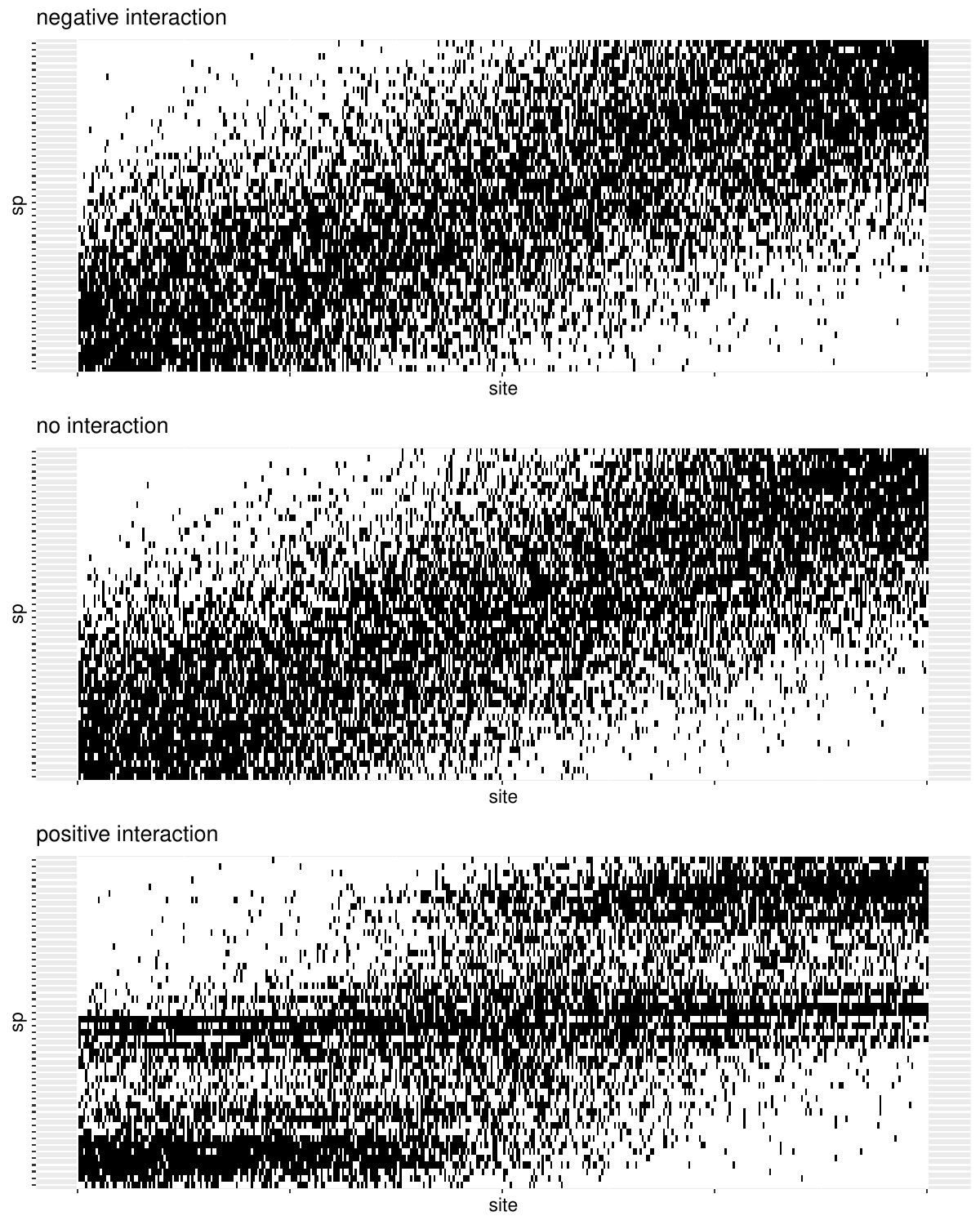}
		\caption{Presence-absence of species ($y$-axis) along the environmental gradient ($x$-axis) for the VirtualCom simulations across the three interaction scenarios.}
		\label{fig:spp_distrib_VC}
	\end{center}
\end{figure}

\begin{figure}[H]
	\begin{center}
		\includegraphics
		[width=0.6\textwidth]{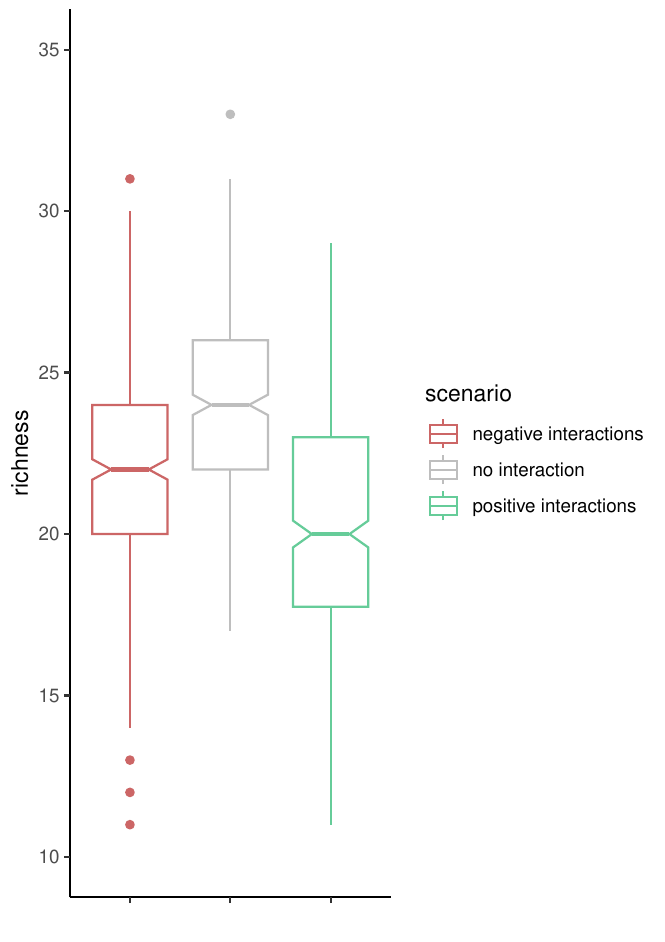}
		\caption{Distribution of species richness (observed number of present species) for the VirtualCom  simulation under the three interaction scenarios.}
		\label{fig:simu_VC_richness}
	\end{center}
\end{figure}

\begin{figure}[H]
	\begin{center}
		\includegraphics
		[width=0.6\textwidth]{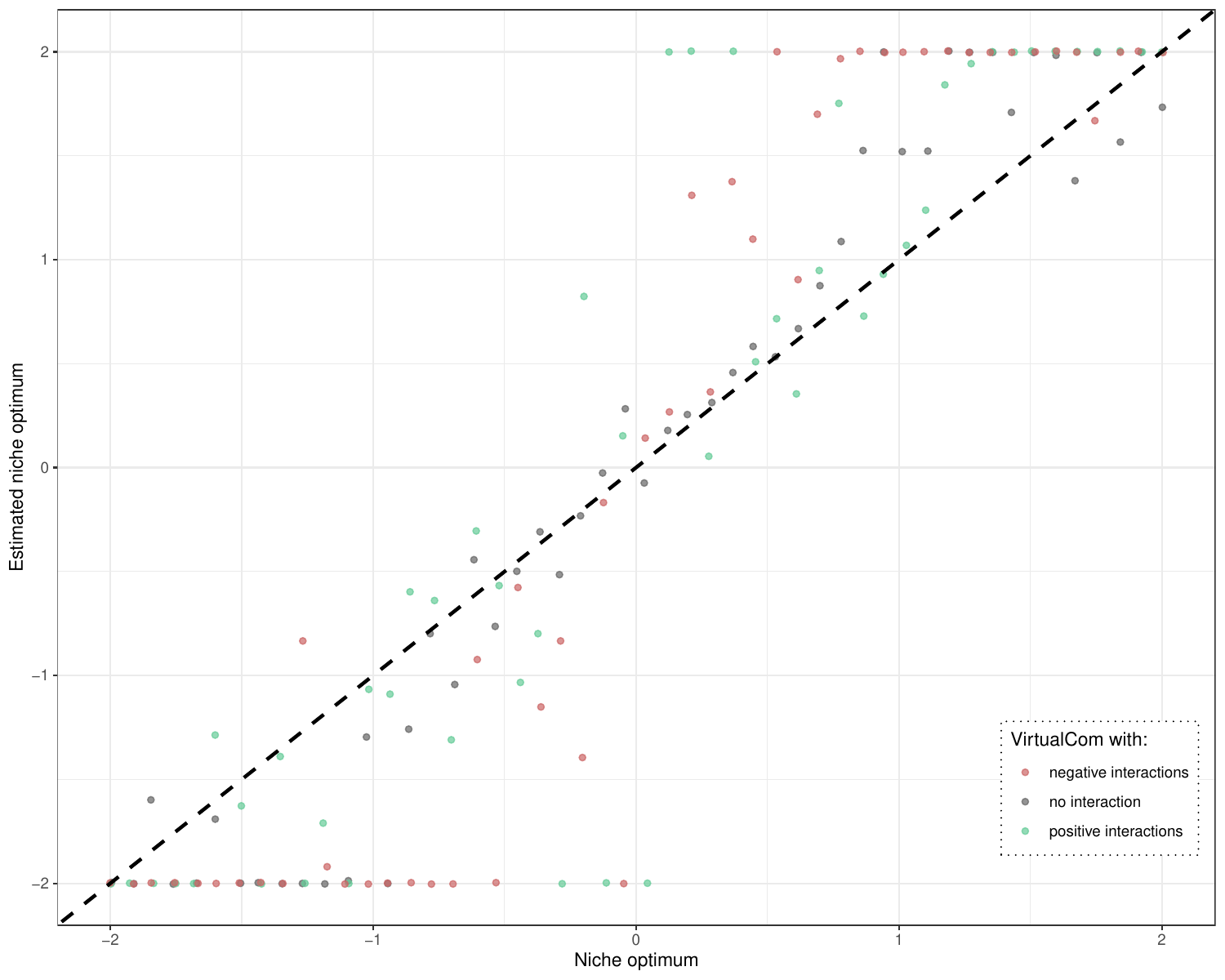}
		\caption{Estimated niche optima versus true niche optima for the VirtualCom simulation under the three interaction scenarios.}
		\label{fig:simu_VC_fit_niches}
	\end{center}
\end{figure}

\subsubsection*{Results and discussion}

We notice here that species richness is lower in the facilitation case than in the other cases. This might seem counter intuitive, but comes from the constraint of VirtualCom of keeping fixed the number of individual at each site. Therefore, positive interactions tend to produce communities with a lower number of species, since the few species that facilitate each other and that can survive at the given environmental conditions keep enhancing their probability of presence and cannot be replaced by other species. Instead, the cases with negative interactions, or without interactions, reduces the probability of competitive species, thus favoring all the other species to replace them, leading to an overall higher richness.
We see that ELGRIN reasonably infers the niche parameters and association parameters ($\betap$ and $\betaa$, see Fig.4 in the main text) on this community data built from VirtualCom model (see the discussion in the main text for further insights). We however acknowledge a large variance on association parameters for the negative interactions model, which could be due to the fact that  the VirtualCom model does not express as a ELGRIN one. 
In the no  interactions scenario, we correctly infer that the association parameters under ELGRIN model are estimated around zero.

\subsection{Kolmogorov-Smirnov tests on association parameters}
\blue{To quantitatively investigate the difference between $\betap$ and $\betaa$ distributions in the three simulations, we performed Kolmogorov-Smirnov tests. For each simulation, we tested whether $\betap$ and $\betaa$ distributions were significantly greater (resp. lower) in the scenarios with positive interactions  (resp. negative) from the scenario without interactions. Namely, denoting $\beta^{\text{pos}}$ (resp. $\beta^{\text{no int}}$ and $\beta^{\text{neg}}$) the $\beta$ values inferred under the positive interaction scenario (resp. no interactions and negative interactions) and $F_{\beta^{\text{pos}}}$ (resp. $F_{\beta^{\text{no int}}}$ and $F_{\beta^{\text{neg}}}$) the corresponding cdfs,
we tested in the positive scenario the null hypothesis $H_0:$ $F_{\beta^{\text{pos}}}= F_{\beta^{\text{no int}}}$  against the alternative $H_1:$ $F_{\beta^{\text{pos}}} \ge F_{\beta^{\text{no int}}}$ (this corresponds to stochastic ordering). In the same way, for the negative interaction scenario, we tested  $H_0:$ $F_{\beta^{\text{neg}}}= F_{\beta^{\text{no int}}}$  against the alternative $H_1:$ $F_{\beta^{\text{neg}}} \le F_{\beta^{\text{no int}}}$.
We recall the concept of stochastic ordering for 2 random variables $U,V$: we have that `$U$ is stochastically larger than $V$`, denoted $U \succeq V$ when the corresponding cdfs satisfy $F_U\ge F_V$ which in practice corresponds to the fact that a random observation from $U$ `tends to be larger` than one from $V$.  

In the three simulations, the tests correctly identify significant differences between interactions and no interaction scenarios (Table \ref{tab:KStest}). For the Lotka-Volterra simulation, the $p$-values for the comparison between the positive interaction scenario and the no-interaction scenario were slightly greater than the $p$-values of the other tests, in accordance to the qualitative assessment of the simulation results.}

\begin{table}[]
    \centering
    \blue{
\begin{tabular}{ |p{3cm}||p{6cm}|p{6cm}|  }
 \hline
 \multicolumn{3}{|c|}{Results of Kolmogorov-Smirnov tests} \\
 \hline
 Simulation and parameter & \small{$p$-value of the test $H_0:$ $F_{\beta^{\text{pos}}}= F_{\beta^{\text{no int}}}$  against $H_1:$ $F_{\beta^{\text{pos}}} \ge F_{\beta^{\text{no int}}}$
 } 
 & \small{$p$-value of the test $H_0:$ $F_{\beta^{\text{neg}}}= F_{\beta^{\text{no int}}}$  against $H_1:$ $F_{\beta^{\text{neg}}} \le F_{\beta^{\text{no int}}}$
 }\\
 \hline
 LV, $\betap$   & $0.0020$    &$< 2.2\mathrm{e}{-16}$\\
LV,  $\betaa$   & $0.0040$    &$< 2.2\mathrm{e}{-16}$\\ \hline
 CE, $\betap$   & $< 2.2\mathrm{e}{-16}$    & $< 2.2\mathrm{e}{-16}$\\
CE  $\betaa$   & $< 2.2\mathrm{e}{-16}$    & $< 2.2\mathrm{e}{-16}$\\ \hline
 VC, $\betap$   & $7.8\mathrm{e}{-16}$ & $< 2.2\mathrm{e}{-16}$ \\
VC,  $\betaa$   & $< 2.2\mathrm{e}{-16}$ & $< 2.2\mathrm{e}{-16}$ \\ \hline
\end{tabular}
    \caption{Results of Kolmogorov-Smirnov tests on association parameters ($\betaa$ and $\betap$) in the three simulations settings: Lotka-Volterra (LV), Colonisation-extinction (CE) and VirtualCom (VC). The tests compare the distribution of association parameters between the positive interactions and the no interaction scenario (second column) and also between the negative interactions scenario and the no interaction scenario (third column).}
    \label{tab:KStest}    
}
\end{table}

\section{ Simulation beyond model assumptions}
\label{sec_SI:LVintra}
We provide here a test of our model when species communities are simulated based on processes that are not accounted for by ELGRIN. In particular, we take the example of Lotka-Volterra models
(see Section \ref{SI_sec:LV}) where intraspecific interactions are higher than interspecific ones. ELGRIN does not account for these intraspecific interactions (i.e., it does not model self-loops), and we might thus expect that it will struggle in correctly retrieving model parameters. An HTLM vignette for this simulation is available at \url{https://plmlab.math.cnrs.fr/econetproject/econetwork}.

\subsection{Simulation set-up} 
We set-up simulations accordingly to Section \ref{SI_sec:LV}. We simulate three different scenarios: positive (i.e. mutualism), negative (i.e.
competition) or no interactions, using the same  interaction network, niche optima (Fig. \ref{fig:niches_net_LV}) and simulation parameters. However, we increase the intraspecific competition term (i.e., parameter $c$ in Equation~\eqref{lvEq}), from $1/10$ to $2$, in order to make it stronger than interspecific interactions.
The distribution of species presence-absences and species richness under the three interaction scenarios is represented in Fig.~\ref{fig:spp_distrib_LV_intra} and Fig.~\ref{fig:simu_LV_richness_intra}. Niche optima inferred from ELGRIN on this dataset are shown in Fig.~\ref{fig:simu_LV_fit_niches_intra}. Inferred association parameters $\betaa$ and $\betap$ are represented in Fig.~\ref{fig:simu_LV_beta_intra}.

\begin{figure}[H]
	\begin{center}
		\includegraphics
		[width=0.6\textwidth]{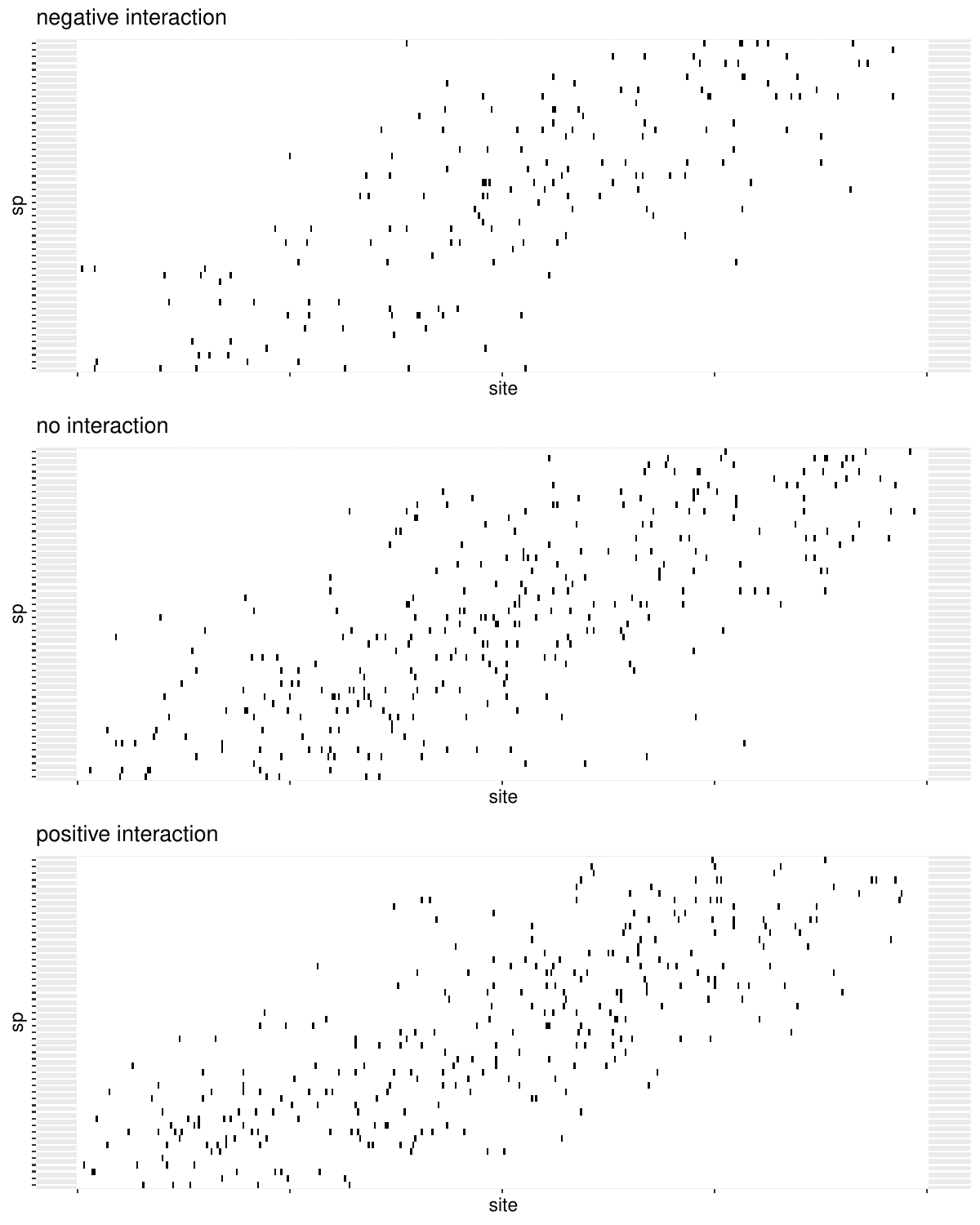}
		\caption{Presence-absence of species ($y$-axis) along the environmental gradient ($x$-axis) for the Lotka-Volterra simulations with intraspecific interactions larger than interspecific ones, across the three interaction scenarios.}
		\label{fig:spp_distrib_LV_intra}
	\end{center}
\end{figure}

\begin{figure}[H]
	\begin{center}
		\includegraphics
		[width=0.6\textwidth]{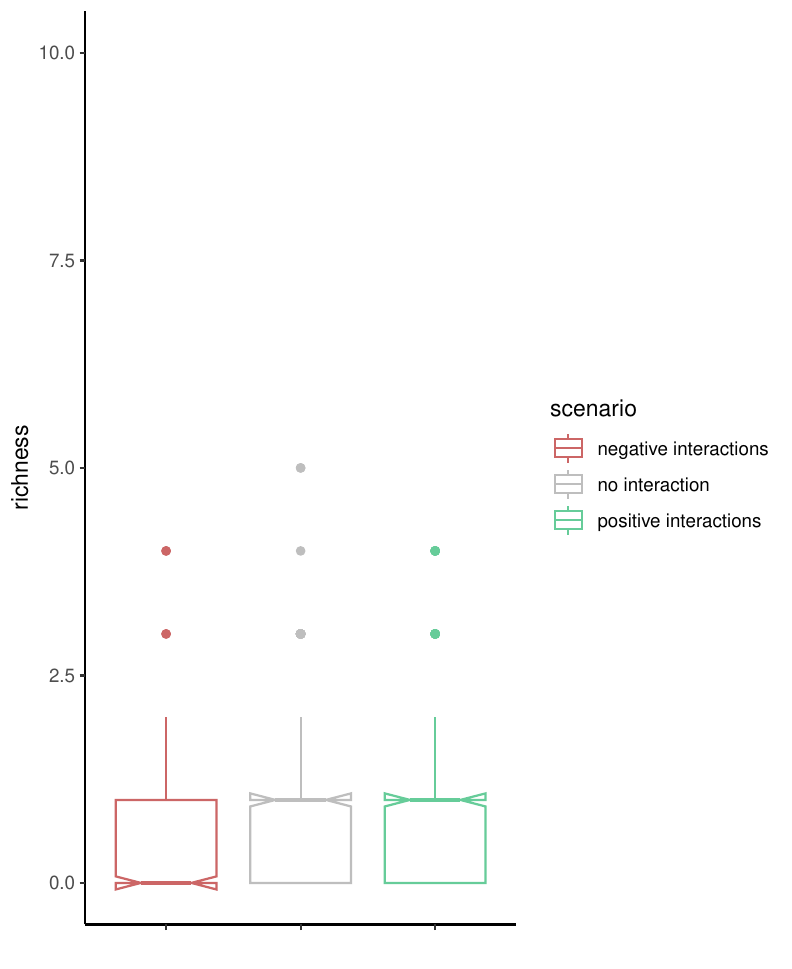}
		\caption{Distribution of species richness (observed number of present species) for the Lotka-Volterra simulation with intraspecific interactions larger than interspecific ones,  under the three interaction scenarios.}
		\label{fig:simu_LV_richness_intra}
	\end{center}
\end{figure}

\begin{figure}[H]
	\begin{center}
		\includegraphics
		[width=0.6\textwidth]{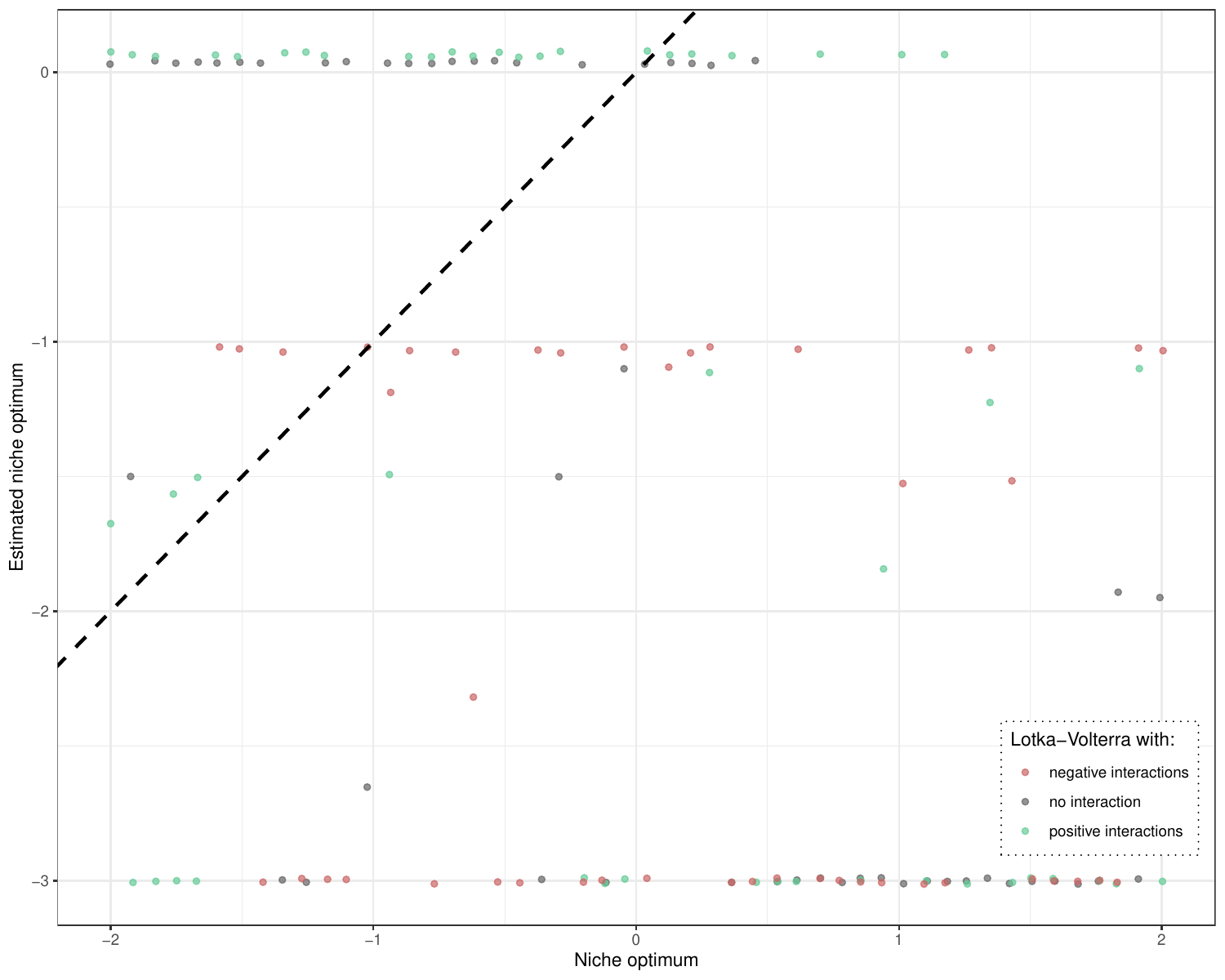}
		\caption{Estimated niche optima versus true niche optima for the Lotka-Volterra simulation with intraspecific interactions larger than interspecific ones, under the three interaction scenarios.}
		\label{fig:simu_LV_fit_niches_intra}
	\end{center}
\end{figure}

\begin{figure}[H]
	\begin{center}
		\includegraphics
		[width=0.6\textwidth]{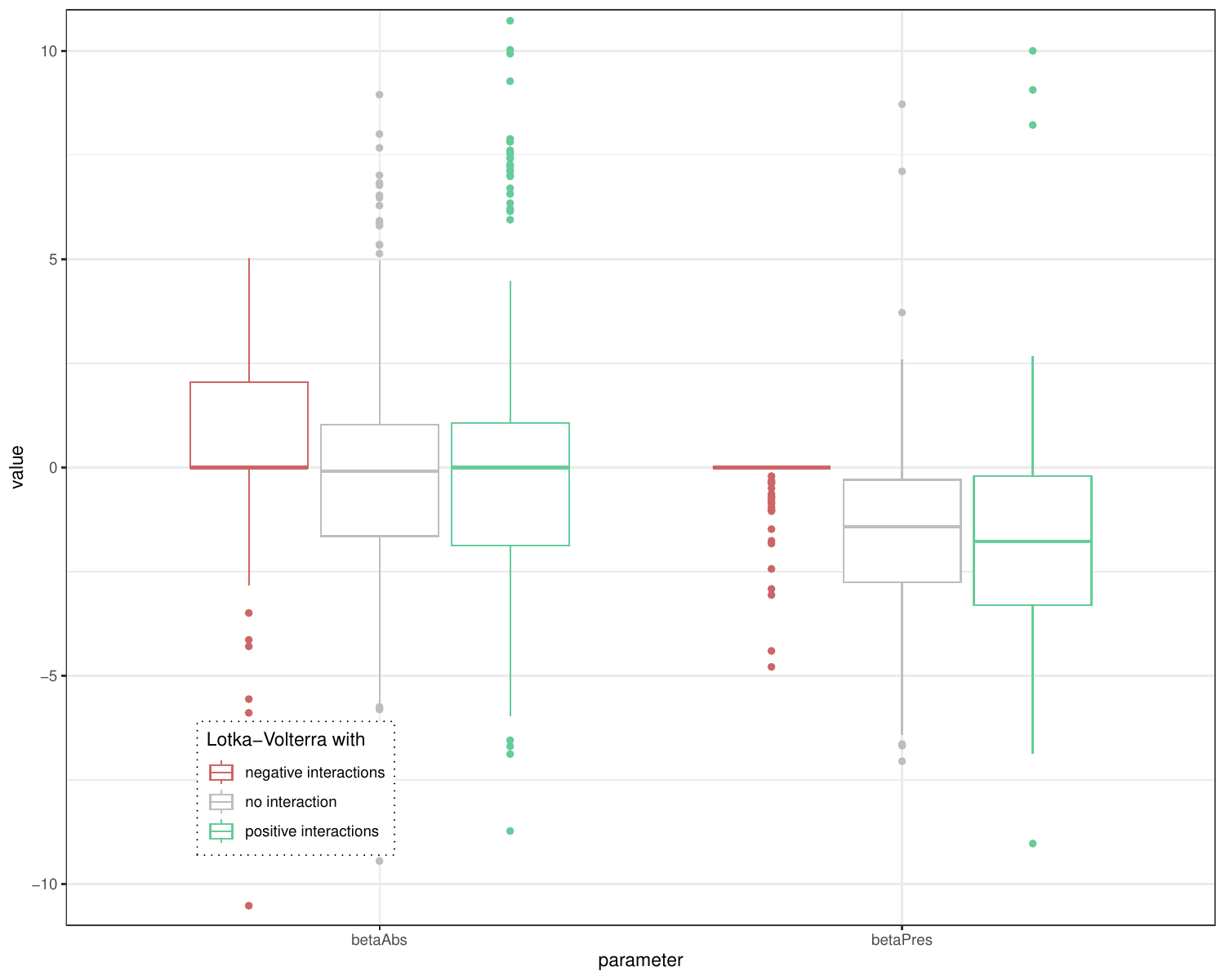}
		\caption{Distribution of co-absence $\betaa$  and co-presence $\betap$ strengths inferred using ELGRIN on simulated ecological communities using a Lotka-Volterra model  with intraspecific interactions larger than interspecific ones, under the three interaction scenarios.}
		\label{fig:simu_LV_beta_intra}
	\end{center}
\end{figure}

\subsubsection*{Results and discussion}

We notice that mean species richness increases from the competition to the mutualistic scenarios, as positive interactions enhance the possibility of species to be present (vice-versa for competition). However, the community matrices are very sparse and species richness is overall very low for the three different scenarios (Fig.~\ref{fig:simu_LV_richness_intra}).
Overall, ELGRIN does not correctly infer model parameters on this community data built from Lotka-Volterra model with large intraspecific interactions. Niche optima are badly estimated (Fig~\ref{fig:simu_LV_fit_niches_intra}) and the inferred association parameters do not show the expected patterns ($\betaa$ and $\betap$, as shown in Fig.~\ref{fig:simu_LV_beta_intra}). We see that ELGRIN cannot correctly disentangle between the three different simulated scenarios. Indeed, we might expect $\betaa$ and $\betap$ parameters to be negative in the negative interaction case, positive in the positive interaction one, and around zero in the no-interaction case. This is generally not the case here, where the inferred $\beta$ parameters are generally close to zero and are lower for the positive interactions scenario than for the negative one.
The poor performance of ELGRIN when intraspecific interactions are higher than interspecific ones is not surprising. As discussed in the main text it is possible to simulate species distributions on which ELGRIN will fail in recovering the true underlying generation process, because these datasets simply do not show anymore enough information about the process that generated them, and could be the result of a completely different scenario, in particular the one inferred by ELGRIN. As such, when using ELGRIN - or any other statistical model - we must bear in mind its model assumptions, knowing that inference might be blurred when other processes are at play.
\section{ Empirical case study}
\subsection{ Relation between $\betap$ and $\betaa$}
\label{SI_sec:realdata1}

Fig.~\ref{fig:supprelation} shows the correlation between the values $\betap$ and $\betaa$ estimated through ELGRIN on the European tetrapods case study.

\begin{figure}[H]
	\begin{center}
	\includegraphics[height=6cm]{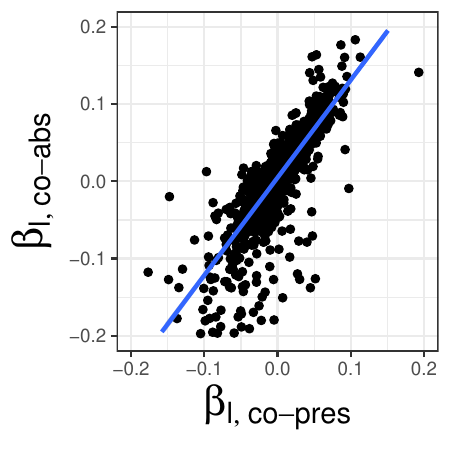}
	{
	\caption{Results of ELGRIN on the European tetrapods case study. The parameters $\betap$ and $\betaa$ were highly correlated.}\label{fig:supprelation}}
	\end{center}
\end{figure}

\subsection{Kolmogorov-Smirnov tests on the distributions}
\label{SI_sec:realdata2}
\blue{We performed the following tests on the $\betap$ parameters estimated from the data:
\begin{itemize}
    \item Denoting $F_{\beta^{\text{high alt}}}$ (resp. $F_{\beta^{\text{low alt}}}$) the cdf of the $\betap$ values inferred at locations with altitude above 1600m (resp. below 1600m), we tested the null hypothesis $H_0: F_{\beta^{\text{high alt}}}= F_{\beta^{\text{low alt}}}$ against the alternative  $H_1: F_{\beta^{\text{high alt}}} \le  F_{\beta^{\text{low alt}}}$. The resulting $p$-value is inferior to $2.2\mathrm{e}{-16}$.
    \item Denoting $F_{\beta^{\text{high richness}}}$ (resp. $F_{\beta^{\text{low richness}}}$) the cdf of the $\betap$ values inferred at locations with richness larger than 200,  we tested the null hypothesis $H_0: F_{\beta^{\text{high richness}}}= F_{\beta^{\text{low richness}}}$ against the alternative  $H_1: F_{\beta^{\text{high richness}}} \ge  F_{\beta^{\text{low richness}}}$. The resulting $p$-value is inferior to $2.2\mathrm{e}{-16}$.
    \item Denoting $F_{|\beta^{\text{high connect}}|}$ (resp. $F_{|\beta^{\text{low connect}}|}$) the cdf of the $|\betap|$ values inferred at locations with connectance larger than its median value ($0.062$),  we tested the null hypothesis $H_0: F_{|\beta^{\text{high connect}}|}= F_{|\beta^{\text{low connect}}|}$ against the alternative  $H_1: F_{|\beta^{\text{high connect}}|} \ge  F_{|\beta^{\text{low connect}}|}$. The resulting $p$-value is inferior to $2.2\mathrm{e}{-16}$.
\end{itemize}
}

\renewcommand\refname{References for Appendix}

\end{document}